\documentclass[12pt,a4paper]{article}

\usepackage{amsmath,amssymb,stmaryrd,latexsym,euscript,array,amsthm}
\usepackage{fullpage,authblk}

\usepackage{microtype}
\usepackage{mathrsfs}

\usepackage{url}
\usepackage{hyperref} 

\usepackage[ruled,noend,linesnumbered]{algorithm2e}

\usepackage{tikz}
\usetikzlibrary{calc}
\usetikzlibrary{patterns}
\usepgflibrary{shapes.misc} 
\usepgflibrary{patterns}

\newcommand{\suchthat}{\,\big|\,}
\usepackage{pifont}

\newcommand{\refiner}{\mathit{ref}}

\theoremstyle{plain}

\newtheorem{definition}{Definition}

\newtheorem{theorem}[definition]{Theorem}
\newtheorem{lemma}[definition]{Lemma}

\newtheorem{proposition}[definition]{Proposition}
\theoremstyle{remark}
\newtheorem*{remark}{Remark}

\DeclareMathOperator{\remove}{Remove}
\DeclareMathOperator{\notRel}{NotRel}
\DeclareMathOperator{\pre}{pre}

\DeclareMathOperator{\refine}{Refine}
\DeclareMathOperator{\initRefine}{InitRefine}
\DeclareMathOperator{\splitRefine}{SplitRefine}

\DeclareMathOperator{\delete}{Delete}
\DeclareMathOperator{\splitDelete}{SplitDelete}

\DeclareMathOperator{\dom}{dom}
\DeclareMathOperator{\Copy}{copy}
\DeclareMathOperator{\SL}{sl}

\SetKwData{NotInRemove}{notInRemove} 
\SetKwData{InRemove}{inRemove}
\SetKwData{Rel}{Rel}
\SetKwData{RelCount}{RelCount}
\SetKwData{SplitCount}{splitCount}
\SetKwData{NotRel}{NotRel}
\SetKwData{Block}{block}
\SetKwData{State}{state}
\SetKwData{Node}{node}
\SetKwData{True}{true}
\SetKwData{False}{false}
\SetKwData{Mark}{mark}
\SetKwData{Count}{count}
\SetKwData{PreB}{PreB}
\SetKwData{Post}{post}
\SetKwData{Pre}{pre}
\SetKwData{PreRel}{PreRel}
\SetKwData{Remove}{Remove}
\SetKwData{Src}{src}
\SetKwData{Dst}{dst}
\SetKwData{Sl}{sl}
\SetKwData{Label}{label}
\SetKwData{Range}{range}
\SetKwData{PosQp}{posQp}
\SetKwData{Idx}{Idx$_P$}
\SetKwData{Start}{start}
\SetKwData{End}{end}
\SetKwData{Index}{index}

\SetKwFunction{Put}{put}
\SetKwFunction{Get}{get}
\SetKwFunction{Contains}{contains}
\SetKwFunction{Newnode}{newNode()}
\SetKwFunction{Newblock}{newBlock()}
\SetKwFunction{IsEmpty}{is\_empty()}
\SetKwFunction{States}{states}
\SetKwFunction{MoveToHead}{moveToHead}
\SetKwFunction{Add}{add}
\SetKwFunction{MoveHeadTo}{moveHeadTo}
\SetKwFunction{Init}{Init}
\SetKwFunction{Split}{Split}

\SetKwFunction{Take}{take}
\SetKwFunction{Add}{add}
\SetKwFunction{Swap}{swap}  

\SetKwInput{Assert}{Assert}

\title{Three Simulation Algorithms for Labelled Transition Systems}

\author{G\'erard C\'ec\'e}
\affil{\small{}FEMTO-ST, UMR 6174,\\
  1 cours Leprince-Ringuet,  BP 21126,\\
  25201 Montbéliard Cedex France\\    
  \texttt{Gerard.Cece@femto-st.fr}}

\begin{document}

\maketitle

\begin{abstract}
Algorithms which compute the coarsest simulation preorder are generally 
designed on Kripke structures. Only in a second time they are extended to
labelled transition systems. By doing this, the size of the alphabet
appears in general as a multiplicative factor to both time and space
complexities. Let $Q$ denotes the state space, $\rightarrow$ the transition
relation, $\Sigma$ the alphabet and $P_{sim}$ the partition of $Q$ induced
by the coarsest simulation equivalence. In this paper, we propose a base 
algorithm which minimizes, since the first stages of its design, the
incidence of the size of the alphabet in both time and space
complexities. This base algorithm, inspired by the one of Paige and Tarjan in
1987 for bisimulation and the one of Ranzato and Tapparo in 2010 for
simulation, is then derived in three versions. One of them has the best bit
space complexity up to now,
$O(|P_{sim}|^2+|{\rightarrow}|.\log|{\rightarrow}|)$, 
while another one has the best time complexity up to now,
$O(|P_{sim}|.|{\rightarrow}|)$. Note the absence of the alphabet in these
complexities. A third version happens to be a nice compromise between space
and time since it runs in $O(b.|P_{sim}|.|{\rightarrow}|)$ time, with $b$ a
branching factor generally far below $|P_{sim}|$, and uses
$O(|P_{sim}|^2.\log|P_{sim}|+|{\rightarrow}|.\log|{\rightarrow}|)$ bits.
\end{abstract}

\section{Introduction}
\label{sec:introduction}

Simulation is a behavioral relation between processes \cite{Milner71}. It is
mainly used to tackle the state-explosion problem that arises in \emph{model
checking} \cite{GPP03,ABH+08} and to speed up the test of inclusion of languages
\cite{ACH+10}. It can also be used as a sufficient condition for the inclusion 
of languages when this test is undecidable in general \cite{CG11}. The paper
\cite{GPP03} gives a complete state of the art of the notion. 

\subsection{Last Ten Years}
\label{sec:lastYears}

Let $T=(Q,\Sigma,\rightarrow)$ be a Labelled Transition System (LTS) with
$Q$ its set of states, $\Sigma$ its alphabet and
$\rightarrow\subseteq Q\times\Sigma\times Q$ its transition relation. A
relation $\mathscr{S}\subseteq Q\times Q$ is a simulation over $T$ if for any
transition $q_1\xrightarrow{a}q'_1$ and any state $q_2\in Q$ such that
$(q_1,q_2)\in \mathscr{S}$, there is a transition $q_2\xrightarrow{a}q'_2$
such that $(q'_1,q'_2)\in \mathscr{S}$. The simulation $\mathscr{S}$ is a
bisimulation if $\mathscr{S}^{-1}$ is also a simulation. 
Given any preorder (reflexive and
transitive relation) $\mathscr{R}\subseteq Q\times Q$ the
purpose of this paper is to design efficient algorithms which compute the coarsest
simulation over $T$ included in $\mathscr{R}$.

In the context of Kripke structures, which are transition systems where
only states are labelled, the most efficient algorithms are GPP, the one of 
Gentilini, Piazza and Policriti \cite{GPP03} (corrected by van Glabbeek and 
Ploeger \cite{GP08}), for the space efficiency, and RT, the one of Ranzato and
Tapparo \cite{RT10}, for the time efficiency. These two algorithms either
use, for GPP, or extend, for RT, HHK the one of Henzinger,
Henzinger and Kopke \cite{HHK95}.

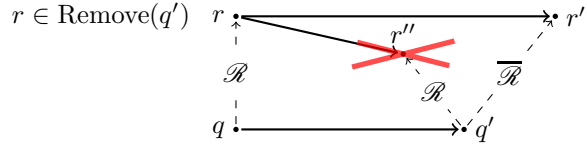
  \begin{center}
    \begin{tikzpicture}[shorten >=2pt, shorten <=2pt,font=\footnotesize]
      \path coordinate [label=right:\textcolor{black}{$q'$}] (q') [fill]
      circle (1pt) ;

      \path ++(0.2,0) ++(90:1.5cm) +(1,0) coordinate (r') [fill] circle (1pt)
      ; \path (r') coordinate[label=right:\textcolor{black}{$r'$}] ;
  
      \path (q') ++(90:1.5cm) ++(-3,0) coordinate (r) ++(-.2,0)
      node[circle,minimum size=.5cm, label={left:$r\in \remove(q')$}] {};

      \path (r) coordinate[label=left:\textcolor{black}{$r$}] [fill] circle
      (1pt) ;
      \path ++(-0.8,1)
      coordinate[label=above:\textcolor{black}{$r''$}] (r'')
      [fill] circle (1pt) ;

      \path let \p{r}=(r),\p{q'}=(q'.center) in (\x{r},\y{q'}) coordinate
      (q'') ; \path (q'') coordinate [label=left:\textcolor{black}{$q$}]
      (q) [fill] circle (1pt) ;
 
      \path[every edge/.style={->,dashed,draw},circle,inner sep=2pt, every node/.style={fill=white}]
      (q') edge node {$\mathscr{R}$} (r'') edge node
      {$\overline{\mathscr{R}}$} (r') (q) edge node {$\mathscr{R}$} (r) ;
      
      \path[->,auto,circle,inner sep=1pt,thick] (r) edge (r') edge (r'')
      (q) edge (q') ;
        
      \node [cross out,draw=red,line width=2pt,draw opacity=.70,text
      width=1cm] at (r'') {};
    \end{tikzpicture}
  \end{center}

The starting idea of HHK, see the above figure, is to
consider couples $(q', r')$ that do not belong to $\mathscr{R}$ (thus
$(q',r')$ belongs to $\overline{\mathscr{R}}$ the
complement of $\mathscr{R}$) and to
propagate this knowledge backward by refining $\mathscr{R}$. For each state
$q'$ a set of states, $\remove(q')$, is maintained. This set is
included in the complement of 
$\pre(\mathscr{R}(q'))$, the set of states which have at least one outgoing
transition leading to a state related to $q'$ by $\mathscr{R}$.
In the figure above, to illustrate that a state $r$ belongs to
$\remove(q')$ we depict that there is no state $r''$ reachable from $r$ and
such that $(q',r'')$ belongs to $\mathscr{R}$.
For a given state $q'$, $\remove(q')$  is used as
follows: for each couple $(q,r)\in \pre(q')\times\remove(q')$, with
$\pre(q')$ the set of states leading, by a transition, to $q'$, if $(q,r)$
belongs to $\mathscr{R}$ then it is removed and $\remove(q)$ is possibly
updated. The couple $(q,r)$ is safely removed from $\mathscr{R}$ because by
the definition of $r\in\remove(q')$ it is impossible that $(q,r)$ belongs to a
simulation included in $\mathscr{R}$. The algorithm HHK runs in
(remember, for the moment transitions are not labelled) $O(|{\rightarrow}|.|Q|)$ time and uses
$O(|Q|^2.\log|Q|)$ bits for all the $\remove$ sets. Note that in order to
achieve the announced time complexity the authors use a set of counters
which plays the same role as this introduced by Paige and Tarjan
\cite{PT87} to lower the time complexity for the corresponding bisimulation
problem from $O(|{\rightarrow}|.|Q|)$ to $O(|{\rightarrow}|.\log|Q|)$. In 
HHK the set of counters enable to lower the time complexity for
the simulation problem from $O(|{\rightarrow}|.|Q|^2)$ to
$O(|{\rightarrow}|.|Q|)$.

If one extends HHK to LTS, where transitions are
labelled, there is a necessity to maintain a $\remove$ set for each couple
state-letter $(q',a)$ because, now, $\remove_a(q')$ is included in the
complement of 
$\pre_a(\mathscr{R}(q'))$ and $\pre$ need to depend on the letters
labelling the transitions. Then, any natural extension of HHK to
LTS uses $O(|\Sigma|.|Q|^2.\log|Q|)$ bits for
all the $\remove$'s.

Let us come back to Kripke structures. The main difference between
HHK in one hand and, GPP  and 
RT in the other hand is that the last two do not encode the current relation
$\mathscr{R}$ by a binary matrix of size $|Q|^2$ but by a
partition-relation pair: a couple $(P,R)$
with $R$ a binary matrix of size $|P|^2$ and $P$ the partition of $Q$
issued from the equivalence relation $\mathscr{R}\cap\mathscr{R}^{-1}$.
The difficulty of the proofs and the abstract interpretation framework
put aside, RT is a thus a direct  
reformulation of HHK but with partition-relation pairs instead of
mere relations between states. Note that in order to have sound refinements
of the relation $R$, blocks of $P$ may first be split at each main
iteration of the algorithm. The algorithm 
RT maintains for each block $B\in P$ a set $\remove(B)$ included
in the complement of $\pre(\mathscr{R}(B))$, the set of states which have
at least one outgoing transition leading to the set of states related to $B$
by $\mathscr{R}$. The algorithm  runs in
 $O(|P_{sim}|.|{\rightarrow}|)$ time and uses
$O(|P_{sim}|.|Q|.\log|Q|)$ bits for all the $\remove$'s, with $P_{sim}$ the
partition associated to the coarsest simulation relation included in the
initial preorder $\mathscr{R}_{init}$. In \cite{CRT11},  Crafa and the
authors of RT, reduced this space complexity to
$O(|P_{sim}|.|P_{sp}|.\log|P_{sp}|)$ with $P_{sp}$ a partition whose size
is between these of $P_{sim}$
and $P_{bis}$, the partition associated to the coarsest \underline{bi}simulation
included in $\mathscr{R}_{init}\cap\mathscr{R}^{-1}_{init}$. The goal,
which has not been achieved, was a
bit space complexity of $O(|P_{sim}|^2.\log|P_{sim}|)$.
The bit space complexity of \cite{CRT11} is achieved at the cost of an increase of the time
complexity comparable with $O(|P_{sim}|^2.|{\rightarrow}|)$. The algorithm
GPP uses a partition-relation pair 
$(P,R)$ too. It also proceed by iterative steps of one split of $P$ and one
refinement
of $R$. A split step is done in a more global way than in RT,
then a refinement step uses HHK on an abstract structure whose
states are blocks of $P$. A refinement step is thus done in
$O(|P_{sim}|.|{\rightarrow}|)$ time (remember, here states are blocks
of $P$). As they prove it, there is at most $|P_{sim}|$ refinement
steps. The entire algorithm is thus done in $O(|P_{sim}|^2.|{\rightarrow}|)$
time. Since states are blocks in this use of HHK, the encoding of
all the $\remove$'s uses $O(|P_{sim}|^2.\log|P_{sim}|)$ bits, which has not
been taken into account in the announced bit space complexity of GPP:
$O(|P_{sim}|^2 + |Q|.\log|P_{sim}|)$.

The paper \cite{ABH+08} provides an adaptation for LTS
of RT. It is also a very useful translation of
RT from the context of abstract interpretation to a more classical
algorithmic view on simulations. The algorithm of \cite{ABH+08} runs in
$O(|\Sigma|.|P_{sim}|.|Q|+|P_{sim}|.|{\rightarrow}|)$ time and uses
$O(|\Sigma|.|P_{sim}|.|Q|.\log|Q|)$ bit space.

\subsection{Our Contributions}
\label{sec:our-contributions}

We have mainly focused our attention on the $|\Sigma|$ factor which is
present in both 
time and space complexities of the simulation algorithm in
\cite{ABH+08}. The major step was to realize that if the 
$\remove_a(B)$ set associated with a block $B\in P$ need to depend on a
letter, a set of blocks not related to $R(B)=\{C\in P\suchthat
(B,C)\in R\}$ does not depend on any letter. Therefore, instead of
maintaining $\remove_a(B)$ we maintain $\notRel(B)$ a set of
blocks included in the complement of $R(B)$ and we compute $\remove_a(B)$
only when we need it. Therefore, we do not have to store it. The great by-product
of doing this is that, for each block, we now maintain a set of blocks,
encoded with $O(|P_{sim}|.\log|P_{sim}|)$ bits, and
not a set of states encoded with $O(|Q|.\log|Q|)$ bits. Thus, we also
achieve the main goal of \cite{CRT11}.

In the next two sections we state the preliminaries and clarify our views
regarding the underlying 
theory. Then, we propose our base algorithm. In the section 
which follows we derive this base algorithm in several versions. The first
one runs in $O(\min(|P_{sim}|,b).|P_{sim}|.|{\rightarrow}|)$ time, with $b$
a branching factor, of the underlying LTS, defined in Section \ref{sec:compromise}, and uses
$O(|P_{sim}|^2.\log|P_{sim}|+|{\rightarrow}|.\log|{\rightarrow}|)$ bit
space. 
By adding a set of counters, like in \cite{RT10,ABH+08}, we obtain a second
version of the algorithm that runs in $O(|P_{sim}|.|{\rightarrow}|)$ time and uses
$O(|P_{sim}|.|\SL(\rightarrow)|.\log|Q|+|{\rightarrow}|.\log|{\rightarrow}|)$
bit space, with (in common cases):
$|P_{sim}|\leq |Q|\leq |\SL(\rightarrow)|\leq|{\rightarrow}|\leq |\Sigma\times Q|$.
The adding space is used to store the counters and is the price to pay to
obtain the best time complexity. This version of the algorithm becomes the
best one, for LTS, regarding time efficiency.
We then explain why GPP does not have a bit space
complexity of $O(|P_{sim}|^2+|Q|.\log|P_{sim}|)$ but at least of 
$O(|P_{sim}|^2.\log|P_{sim}|+|Q|.\log|Q|)$. Then, we propose the third
version of our base algorithm. It runs in
$O(|P_{sim}|^2.|{\rightarrow}|)$ time and 
uses $O(|P_{sim}|^2+|{\rightarrow}|.\log|{\rightarrow}|)$ bits. It is the
best one regarding space complexity. Then, we detail the data structures that
we use. We end the paper by some perspectives including a future work on
bisimulation.  

\section{Preliminaries}
\label{sec:preliminaries}

Let $Q$ be a set of elements. The number of elements of $Q$ is denoted $|Q|$.
A \emph{binary relation} on $Q$ is a subset
of $Q\times Q$.
In the remainder of this paper, we consider only binary
relations, therefore when we write ``relation'' read ``binary relation''. Let
$\mathscr{R}$ be a relation on $Q$. 
For all $q,q'\in Q$, we may write $q\,\mathscr{R}\,q'$, or
$q \mathbin{\tikz[baseline] \draw[dashed,->] (0pt,.5ex) --
  node[font=\footnotesize,fill=white,inner sep=2pt] {$\mathscr{R}$} (6ex,.5ex);}
q'$
in the figures, when $(q,q')\in\mathscr{R}$. 
We define 
$\mathscr{R}(q)\triangleq\{q'\in Q\suchthat q\,\mathscr{R}\,q'\}$ for $q\in Q$ and
$\mathscr{R}(X)\triangleq\cup_{q\in X}\mathscr{R}(q)$ for $X\subseteq
Q$. In the figures, we note
$X \mathbin{\tikz[baseline] \draw[dashed,->] (0pt,.7ex) --
  node[font=\footnotesize,fill=white,inner sep=2pt] {$\mathscr{R}$} (6ex,.7ex);}
Y$ when there is $(q,q')\in X\times Y$ with $q\,\mathscr{R}\,q'$.
The \emph{domain} of $\mathscr{R}$ is
$\dom(\mathscr{R})\triangleq\{q\in Q\suchthat\mathscr{R}(q)\neq\emptyset\}$. 
 The complement of $\mathscr{R}$ is
$\overline{\mathscr{R}}\triangleq\{(x,y)\in Q\times Q\suchthat (x,y)\not\in 
 \mathscr{R}\}$. Let $\mathscr{S}$ be another relation on $Q$, the
composition of $\mathscr{R}$ by $\mathscr{S}$ is 
$\mathscr{S}\mathrel{\circ}\mathscr{R}\triangleq\{(x,y)\in Q\times Q\suchthat
y\in\mathscr{S}(\mathscr{R}(x))\}$.
The relation $\mathscr{R}$ is said \emph{reflexive} if for all $q\in Q$, we
have $q\,\mathscr{R}\,q$. 
The relation $\mathscr{R}$ is said \emph{reflexive on its domain}
if for all $q\in \dom(\mathscr{R})$, we have $q\,\mathscr{R}\,q$.
The relation $\mathscr{R}$ is said \emph{antisymmetric} if
$q\,\mathscr{R}\,q'$ and $q'\,\mathscr{R}\,q$ implies $q=q'$.
The relation $\mathscr{R}$ is said \emph{transitive} if
$\mathscr{R}\mathrel{\circ}\mathscr{R}\subseteq\mathscr{R}$.

A \emph{preorder} is a reflexive and transitive relation. 
Let $X$ be a
set of subsets of $Q$, we note $\cup X\triangleq \cup_{B\in X}B$.
A \emph{partition} of $Q$ is a set of non empty subsets of $Q$, called
\emph{blocks}, that are pairwise disjoint and whose union gives $Q$. A
\emph{partition-relation pair} over $Q$ is a pair $(P,R)$ such that $P$ is
a partition of $Q$ and $R$ is a reflexive relation on $P$. A
partition-relation pair $(P,R)$ is said antisymmetric if its relation $R$ is
antisymmetric. From a partition-relation pair $(P,R)$ over $Q$ we derive a
relation $\mathscr{R}_{(P,R)}$ on $Q$ such that:
$\mathscr{R}_{(P,R)} = \cup_{(B,C)\in R}B\times C$.

\begin{definition}
  \label{def:partionable}
   Let $\mathscr{R}$ be a relation on a set $Q$ such that $\mathscr{R}$ is
   reflexive on its domain.   
   \begin{itemize}
   \item For $q\in Q$, we define
     $[q]_\mathscr{R} \triangleq \{q'\in Q \suchthat q \,\mathscr{R}\, q'
     \wedge q' \,\mathscr{R}\, q\}$,
     for $X\subseteq Q$, we define
     $[X]_\mathscr{R}\triangleq\cup_{q\in X}[q]_{\mathscr{R}}$.
   \item A \emph{block} of $\mathscr{R}$ is a non empty set of states $B$
     such that $B=[q]_\mathscr{R}$ for a $q\in Q$.

   \item $\mathscr{R}$ is said \emph{block-definable}, or \emph{definable by
       blocks} if: 
     $\forall q,q'\in Q\,.\,(q,q')\in\mathscr{R}\Rightarrow
     [q]_\mathscr{R}\times[q']_\mathscr{R}\subseteq\mathscr{R}$.
   \end{itemize}
\end{definition}

Let us remark that
a preorder is reflexive and definable by blocks.
The notion of definability by blocks  will be useful since intermediate
relations of our algorithms will be block-definable, but not necessarily
preorders, even if we start from a preorder and finish with a preorder too.

\begin{remark}
  Let $(P,R)$ be an antisymmetric partition-relation pair over a set $Q$. Then:
  $P=\{[q]_{\mathscr{R}_{(P,R)}}\subseteq Q\suchthat q\in Q\}$ and
    $R=\{([q]_{\mathscr{R}_{(P,R)}},[q']_{\mathscr{R}_{(P,R)}})\subseteq
    P\times P\suchthat    
    q\,\mathscr{R}_{(P,R)}\,q'\}$.
\end{remark}

From a reflexive and block-definable relation $\mathscr{R}$ we derive an
antisymmetric partition-relation 
pair $(P_{\mathscr{R}}, R_{\mathscr{R}})$ such that $P_{\mathscr{R}} =
\{[q]_\mathscr{R}\subseteq Q\suchthat q\in Q\}$ and
$R_{\mathscr{R}}=\{([q]_\mathscr{R},[q']_\mathscr{R})\subseteq
    P\times P\suchthat q\,\mathscr{R}\,q'\}$.

\begin{remark}
  Let $\mathscr{R}$ be a reflexive and block-definable relation on $Q$. Then
    $\mathscr{R} = \bigcup_{(B,C)\in R_{\mathscr{R}}} B\times C$.
\end{remark}

The two preceding remarks imply a duality between reflexive and block-definable
relations, and antisymmetric partition-relation pairs. However, the notion
of block-definable relation is somehow more general in the sense that we
require its reflexibility on its domain, not necessarily on the whole state
space $Q$.

Let $T=(Q,\Sigma,\rightarrow)$ be a triple
such that $Q$ is a finite set of elements called \emph{states}, $\Sigma$ is
an \emph{alphabet}, a finite set of elements called $letters$ or
\emph{labels}, and 
$\rightarrow\subseteq Q\times\Sigma\times Q$ is a \emph{transition
  relation} or \emph{set of transitions}. Then, $T$ is called a
\emph{Labelled Transition System (LTS)}. From 
$T$, given a letter $a\in\Sigma$, we define the two following relations:
$\xrightarrow{a}\mathbin{\triangleq}\{(q,q')\suchthat (q,a,q')\in
\rightarrow\}$ and its reverse
$\pre_{\xrightarrow{a}}\triangleq
\{(q',q)\suchthat (q,a,q')\in\rightarrow\}$. When $\rightarrow$ is clear from
the context, we simply note $\pre_a$ instead of
$\pre_{\xrightarrow{a}}$. For $X,Y\subseteq Q$, we note $X\xrightarrow{a}Y$
to express that $X\cap\pre_a(Y)\neq\emptyset$. By abuse of notation, we
also note $q\xrightarrow{a}Y$ for $\{q\}\xrightarrow{a}Y$.
In the complexity analysis of the
algorithms proposed in this paper, a new notion has emerged, that of
\emph{state-letter}. From $T$ we define the set of state-letters
$\SL(\rightarrow)\triangleq\{(q,a)\in Q\times\Sigma\suchthat
\exists q'\in Q\;.\;q\xrightarrow{a}q'\in\rightarrow\}$. For $(q,a)\in \SL(\rightarrow)$,
we simply note $q_a$ instead of $(q,a)$.
 If $T$ is ``normalized'' (see first paragraph of 
Section~\ref{sec:complexity}) we have: 
\begin{equation}
  \label{eq:StateLetterSmaller}
|Q|\leq |\SL(\rightarrow)|\leq|{\rightarrow}|\leq |\Sigma\times Q|
\end{equation}
It is therefore more interesting to use $\SL(\rightarrow)$ instead of $\Sigma\times
Q$.

The following definition of a simulation happens to be more effective
than the classical one given in the introduction.

\noindent
\begin{minipage}[t]{0.8\linewidth}
  \begin{definition}
    \label{def:sim}
    Let $T=(Q,\Sigma,\rightarrow)$ be a LTS and $\mathscr{S}$ be a relation
    on $Q$. The relation $\mathscr{S}$ is a \emph{simulation} over $T$ if:
    $ \forall
    a\in\Sigma\;.\;\mathscr{S}\mathrel{\circ}\pre_a\subseteq\pre_a\mathrel{\circ}\mathscr{S}$.
    For two states $q,q'\in Q$, we say ``$q'$ simulates $q$'' if there is a
    simulation $\mathscr{S}$ over $Q$ such that $q\,\mathscr{S}\,q'$.
      \end{definition}

\end{minipage}
\hfill
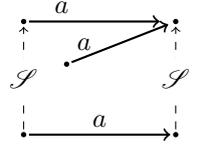
\begin{minipage}[t]{0.15\linewidth}
      \begin{tikzpicture}[baseline=(q3.south),shorten >=2pt,shorten
      <=2pt,font=\footnotesize]

      \path coordinate (q1) [fill] circle (1pt) (-2,0) coordinate (q2)
      [fill] circle (1pt) (-2,1.5) coordinate (q3) [fill] circle (1pt)
      (0,1.5) coordinate (q4) [fill] circle (1pt)
      (q3) 
      +(-45:.8) coordinate (q6) [fill] circle (1pt)
      ;
    
      \path[every edge/.style={->,dashed,draw},circle,inner sep=1pt, every node/.style={fill=white}]
      (q2) edge node (S) {$\mathscr{S}$} (q3) (q1) edge node
      {$\mathscr{S}$} (q4);

      \path[->,auto,circle,inner sep=1pt,thick] (q2) edge node {$a$} (q1)
      (q3)edge[shorten >=6pt] node[near start] (base) {$a$} (q4)
      (q6)edge[shorten >=3pt] node[near start] {$a$} (q4) ;
    \end{tikzpicture}
\end{minipage}

\section{Underlying Theory}
\label{sec:underlyingTh}

The first consequence of the definition of a simulation over a LTS
$T=(Q,\Sigma, \rightarrow)$ is that 
states which have an outgoing transition labelled by a letter $a$ can be
simulated only by states which have at least one outgoing transition
labelled by $a$. The next definition and lemma establish that we can restrict
our problem of finding the coarsest simulation inside a preorder to the
search of  the coarsest simulation inside a preorder $\mathscr{R}$ that
satisfies:
\begin{equation}
 \label{eq:InitRefineRestriction}  
 \forall a\in\Sigma\;.\;\mathscr{R}(\pre_a(Q))\subseteq\pre_a(Q).
\end{equation}

\begin{definition}
    Let $T=(Q,\Sigma, \rightarrow)$ be a LTS and 
   $\mathscr{R}$ be a preorder on $Q$. We
   define $\initRefine(\mathscr{R})\subseteq \mathscr{R}$ such that:
   \begin{displaymath}
    (q,q')\in\initRefine(\mathscr{R}) \Leftrightarrow
    (q,q')\in\mathscr{R}\,\wedge\,
    \forall a\in\Sigma\;(q\in\pre_a(Q) \Rightarrow q'\in\pre_a(Q)).
   \end{displaymath}
\end{definition}

\begin{lemma}
  \label{lem:InitRefine}
     Let $T=(Q,\Sigma, \rightarrow)$ be a LTS and
     $\mathscr{U} = \initRefine(\mathscr{R})$ with
   $\mathscr{R}$ a preorder on $Q$. Then:   
   \begin{enumerate}
   \item $\mathscr{U}$ is a preorder,     
   \item for all simulation $\mathscr{S}$ over $T$:     
   $ \mathscr{S}\, \subseteq \mathscr{R}
     \Rightarrow \mathscr{S}\, \subseteq \mathscr{U}$,
 \item \label{lem:InitRefine:3}
   $\forall X\subseteq Q\,\forall a\in\Sigma\;.\;\mathscr{U}(\pre_a(X))
   \subseteq \pre_a(Q)$.
   \end{enumerate}
\end{lemma}
\begin{proof}\mbox{}
  \begin{enumerate}
\item Since $\mathscr{R}$ is a preorder and thus reflexive, $\mathscr{U}$
  is also trivially
  reflexive. Now, let us suppose $\mathscr{U}$ is not transitive. There are
  three states $q_1,q_2,q_3\in Q$ such 
  that: $q_1\,\mathscr{U}\,q_2\wedge q_2\,\mathscr{U}\,q_3 \wedge \neg\;
  q_1\,\mathscr{U}\,q_3$. From the fact that
  $\mathscr{U}\subseteq\mathscr{R}$ and $\mathscr{R}$ is a preorder, we
  get $q_1\,\mathscr{R}\,q_3$. With $\neg\; q_1\,\mathscr{U}\,q_3$ and the
  definition of $\mathscr{U}$ there is $a\in\Sigma$ such that:
  $q_1\in\pre_a(Q)$ and $q_3\not\in\pre_a(Q)$. But $q_1\in\pre_a(Q)$ and
  $q_1\,\mathscr{U}\,q_2$ implies $q_2\in\pre_a(Q)$. With
  $q_2\,\mathscr{U}\,q_3$ we also get $q_3\in\pre_a(Q)$ which contradicts
   $q_3\not\in\pre_a(Q)$.
 \item If this is not true there are two states $q_1,q_2\in Q$ such
  that: $q_1\,\mathscr{S}\,q_2 \wedge \neg\;q_1\,\mathscr{U}\,q_2$. From
 $\mathscr{U}\subseteq\mathscr{R}$ we get $q_1\,\mathscr{R}\,q_2 $. With
 $\neg\; q_1\,\mathscr{U}\,q_2$ and the definition of $\mathscr{U}$ there
 is $a\in\Sigma$ such that $q_2\not\in\pre_a(Q)$ and
 $q_1\in\pre_a(Q)$. With $q_1\,\mathscr{S}\,q_2$ we get $q_2\in
 \mathscr{S}\mathrel{\circ} \pre_a (Q)$. With the hypothesis that $\mathscr{S}$ is a
 simulation, we get $q_2\in \pre_a\mathrel{\circ} \mathscr{S}(Q)$ and thus
 $q_2\in \pre_a(Q)$, since $\mathscr{S}(Q)\subseteq Q$, which contradicts
 $q_2\not\in\pre_a(Q)$.
 
\item This a direct consequence of the definition of $\mathscr{U}$.
\end{enumerate}
\end{proof}

The main idea to obtain efficient algorithms is to consider relations
between blocks of states and not merely relations between
states. Therefore, we need a characterization of the notion of simulation
expressed over blocks.

\begin{proposition}
  \label{prop:blockSim}
   Let $T=(Q,\Sigma,\rightarrow)$ be a LTS and
  $\mathscr{S}$ be a reflexive and block-definable relation on $Q$. The relation
  $\mathscr{S}$ is a simulation over $T$ if and only if:
  \begin{displaymath}
    \label{eq:I}
    \forall a\in\Sigma\; \forall q\in Q\;.\;
    \mathscr{S}\mathrel{\circ}\pre_a([q]_{\mathscr{S}})\subseteq \pre_a\mathrel{\circ}\mathscr{S}([q]_{\mathscr{S}}).
  \end{displaymath}
\end{proposition}
\begin{proof}
  If $\mathscr{S}$ is a simulation then, by definition, we have for any
  $X\subseteq Q$: $\forall
  a\in\Sigma\;.\;\mathscr{S}\mathrel{\circ}\pre_a(X)\subseteq\pre_a\mathrel{\circ}\mathscr{S}(X)$. This
  inclusion is thus also true for $X=[q]_{\mathscr{S}}$. In the other
  direction, if $\mathscr{S}$ is reflexive and block-definable then for any
  $q\in Q$ we get:
  $q\in [q]_{\mathscr{S}}$ and
  $\mathscr{S}(q)=\mathscr{S}([q]_{\mathscr{S}})$. We thus have:
  \begin{displaymath}
     \mathscr{S}\mathrel{\circ}\pre_a(q)\subseteq
     \mathscr{S}\mathrel{\circ}\pre_a([q]_{\mathscr{S}})\subseteq
     \pre_a\mathrel{\circ}\mathscr{S}([q]_{\mathscr{S}})=
      \pre_a\mathrel{\circ}\mathscr{S}(q)
  \end{displaymath}
  which ends the proof.
\end{proof}

Now, suppose we have a reflexive and block-definable relation
$\mathscr{R}$ and we want to remove from $\mathscr{R}$ all couples $(q,r)$ not
belonging in a simulation included in $\mathscr{R}$. If $\mathscr{R}$ is
not already a simulation, from the last proposition, there are a letter $a$
and a block $B$ of $\mathscr{R}$ such that
$\mathscr{R}\mathrel{\circ}\pre_a(B)\not\subseteq
\pre_a\mathrel{\circ}\mathscr{R}(B)$. But we can assume that $\mathscr{R}$
satisfies \eqref{eq:InitRefineRestriction}. With
$Q=\mathscr{R}(B)\cup\overline{\mathscr{R}}(B)$ we get: 
$\mathscr{R}\mathrel{\circ}\pre_a(B)\subseteq
\pre_a(\mathscr{R}(B)\cup\overline{\mathscr{R}}(B))$. This implies the
existence of a non empty set
$\remove\triangleq\pre_a(\overline{\mathscr{R}}(B))\setminus\pre_a(\mathscr{R}(B))$. Let
$r\in\remove$ and $q\in\pre_a(B)$. If $(q,r)\in\mathscr{R}$ we can safely remove
$(q,r)$ from $\mathscr{R}$. Why? Because, if we had $(q,r)\in\mathscr{S}$
with $\mathscr{S}\subseteq\mathscr{R}$ a simulation, with $q\in\pre_a(B)$,
there would be $q'\in B$ such that $q\in\pre_a(q')$ and thus
$r\in\mathscr{S}\mathrel{\circ}\pre_a(q')$. But $\mathscr{S}$ being a simulation, this
implies $r\in\pre_a\mathrel{\circ}\mathscr{S}(q')$ and thus
$r\in\pre_a\mathrel{\circ}\mathscr{R}(B)$ since $\mathscr{S}\subseteq\mathscr{R}$ and
$q'\in B$. This contradicts $r\in\remove$. 
To sum up, we can safely remove $(q,r)$ from $\mathscr{R}$. But can we
safely remove $C\times D$ from $\mathscr{R}$ with $C$ the block of
$\mathscr{R}$ containing $q$ and $D$ the block of $\mathscr{R}$ containing
$r$? In general, the answer is no. However, the remainder of this section
gives, and justifies, sufficient conditions to do so. 
We begin by the key definition of the paper.

\begin{definition}
 Let $T=(Q,\Sigma, \rightarrow)$ be a LTS and 
   $\mathscr{R}$ be a reflexive and block-definable relation on $Q$. A
   \emph{refiner} of $\mathscr{R}$ is a triple
    $(B,\mathscr{R}_1,\mathscr{R}_2)$ with $\mathscr{R}_1$ and
    $\mathscr{R}_2$ two relations on $Q$  such that $\mathscr{R}_1$ is
    block-definable, $B$ is a block of $\mathscr{R}_1$,   
     $\mathscr{R}(B)\subseteq \mathscr{R}_1(B)$
   and
   \begin{displaymath}
     \forall a\in\Sigma\;.\;
     [\pre_a(\mathscr{R}_1(B)\cup\mathscr{R}_2(B))]_{ \mathscr{R}}
     \cup \mathscr{R}(\pre_a(B))
     \subseteq
     \pre_a(\mathscr{R}_1(B)\cup \mathscr{R}_2(B))
   \end{displaymath}
\end{definition}

Let us fix the intuition for the reader. For the above discussion, we take
$\mathscr{R}_1=\mathscr{R}$ and $\mathscr{R}_2=\overline{\mathscr{R}}$
which allows us to satisfy (under the assumption \eqref{eq:InitRefineRestriction}) all
the conditions of the definition of a refiner. However, if we use
$\mathscr{R}_1$ and $\mathscr{R}_2$ like that, we will obtain algorithms whose
time complexity is $O(|P_{sim}|^2.|{\rightarrow}|)$ for
Kripke structures. To obtain algorithms
in $O(|P_{sim}|.|{\rightarrow}|)$ time, still for Kripke structures, we have to consider in
$\overline{\mathscr{R}}(B)$ only what is needed and thus to keep
$\mathscr{R}_2$ smaller than $\overline{\mathscr{R}}$. The presence of the
relation $\mathscr{R}_1$ is due to the management of the different letters
of the alphabet for LTS.
Note first, that constraining for all the letters in the alphabet the
last condition of the definition of a 
refiner  has made it 
independent of a particular letter. During a main iteration of the
algorithm, we consider a relevant refiner. At this stage
$\mathscr{R}_1=\mathscr{R}$. Gradually, as we consider the letters involved
in the transitions leading to $\mathscr{R}_2(B)$, $\mathscr{R}$ is refined
and thus $\mathscr{R}(B)$ stays included in $\mathscr{R}_1(B)$ but may
becomes smaller than $\mathscr{R}_1(B)$.

The first inclusion of the last condition of a refiner,
$[\pre_a(\mathscr{R}_1(B)\cup\mathscr{R}_2(B))]_{ \mathscr{R}}    
 \subseteq\pre_a(\mathscr{R}_1(B)\cup \mathscr{R}_2(B))$, authorizes to split
 the blocks of $P$ either with $\pre_a(\mathscr{R}_1(B))$ or with
 $\remove_{a,\refiner}=\pre_a(\mathscr{R}_2(B))\setminus
 \pre_a(\mathscr{R}_1(B))$ like in the next definition. The former induces
 algorithms that run in 
 $O(|P_{sim}|^2.|{\rightarrow}|)$ time. The latter authorizes a run in 
$O(|P_{sim}|.|{\rightarrow}|)$ time. Let $\mathscr{R}'$ be the relation
issued from the split. The second inclusion of the last condition of a
refiner,
$\mathscr{R}(\pre_a(B)) \subseteq
 \pre_a(\mathscr{R}_1(B)\cup \mathscr{R}_2(B))$,
 enables to soundly refine $\mathscr{R}'$ by
 $[\pre_a(B)]_{\mathscr{R}'}\times[\remove_{a,\refiner}]_{\mathscr{R}'}$.
 The remainder of the section formalizes this approach.

\begin{definition}
  \label{def:refine}
   Let $T=(Q,\Sigma, \rightarrow)$ be a LTS, 
   $\mathscr{R}$ be a reflexive and block-definable
   relation on $Q$, $\refiner=(B,\mathscr{R}_1,\mathscr{R}_2)$ be a refiner of 
   $\mathscr{R}$ and $a\in\Sigma$ be a letter. We  define:    
 \begin{align*}
     \remove_{a,\refiner} &\triangleq
   \pre_a(\mathscr{R}_2(B))\setminus \pre_a(\mathscr{R}_1(B))\\[2ex]   
     \splitDelete_{a,\refiner}(\mathscr{R}) &\triangleq
   \bigcup_{q\in\remove_{a,\refiner}}
   \begin{array}[t]{l}
     ([q]_{\mathscr{R}}
     \setminus\remove_{a,\refiner})
     \times \\[0cm]
     ([q]_{\mathscr{R}} \cap \remove_{a,\refiner})
   \end{array}\\[2ex] 
     \splitRefine_{a,\refiner}(\mathscr{R}) &\triangleq
   \mathscr{R} \setminus \splitDelete_{a,\refiner}(\mathscr{R}) \\[2ex]  
     \delete_{a,\refiner}(\mathscr{R}) &\triangleq
   \bigcup_{\substack{ q\in \remove_{a,\refiner}\\q'\in\pre_a(B)}}
   \begin{array}[t]{l}
     [q']_{\splitRefine_{a,\refiner}(\mathscr{R})} \times \\[0cm]
     [q]_{\splitRefine_{a,\refiner}(\mathscr{R})}
   \end{array}\\[2ex] 
     \refine_{a,\refiner}(\mathscr{R}) &\triangleq
   \splitRefine_{a,\refiner}(\mathscr{R}) \setminus
   \delete_{a,\refiner}(\mathscr{R}) 
 \end{align*}
\end{definition}

We should now prove that if a simulation $\mathscr{S}$ is included in
$\mathscr{R}$ it is still included in
$\refine_{a,\refiner}(\mathscr{R})$. Unfortunately, we do not know how to do
that. However, if, instead of simply asking $\mathscr{S}$ to be included in 
$\mathscr{R}$, we ask $\mathscr{R}$ to be \emph{$\mathscr{S}$-stable} (see
next definition),
everything works nicely.

\begin{definition}
  \label{def:simStable}
  Let $\mathscr{R}$ and $\mathscr{S}$ be two relations on $Q$. The 
  relation $\mathscr{R}$ is said \emph{$\mathscr{S}$-stable} if
  $ \mathscr{S}\mathrel{\circ}\mathscr{R}\subseteq \mathscr{R}$.
\end{definition}

Obviously, if a reflexive relation $\mathscr{R}$ is $\mathscr{S}$-stable
then $\mathscr{S}$ is included in $\mathscr{R}$. Intuitively, the
$\mathscr{S}$-stability of $\mathscr{R}$ is required by the fact that
$\mathscr{R}$ is no longer supposed to be a preorder but since it contains
a transitive simulation (as we will see, the coarsest simulation included
in a preorder is a preorder and thus transitive) it should be transitive
``with'' that simulation.

\begin{theorem}
  \label{th:Refine}
    Let $T=(Q,\Sigma, \rightarrow)$ be a LTS, $\mathscr{R}$ be a  reflexive and
    block-definable relation on $Q$, $\mathscr{S}$ be a simulation
    over $T$, $a\in\Sigma$ be a letter and
    $\refiner=(B,\mathscr{R}_1,\mathscr{R}_2)$ be a
    refiner of 
   $\mathscr{R}$ such that $\mathscr{R}$ and  $\mathscr{R}_1$ are
   $\mathscr{S}$-stable. Let
   $\mathscr{U}=\refine_{a,\refiner}(\mathscr{R})$. Then, $\mathscr{U}$ is a 
   reflexive, block-definable and $\mathscr{S}$-stable relation. Furthermore,
   we have:
  $
     [\pre_a(\mathscr{R}_1(B))]_{ \mathscr{U}}
     \cup \mathscr{U}(\pre_a(B))
     \subseteq
     \pre_a(\mathscr{R}_1(B))$.
\end{theorem}

For the proof, we first need a lemma.
\begin{lemma}
  \label{lem:ImpSstable}
    Let $T=(Q,\Sigma, \rightarrow)$ be a LTS and, $\mathscr{R}$ and
    $\mathscr{S}$ be two relations on $Q$ such that $\mathscr{S}$ is
    a simulation over $T$ and $\mathscr{R}$ is
    $\mathscr{S}$-stable. Then:
    \begin{displaymath}
      \mathscr{S}\mathrel{\circ}\pre_a\mathrel{\circ}\mathscr{R}\subseteq
      \pre_a\mathrel{\circ}\mathscr{R} 
    \end{displaymath}
    Said otherwise, $\pre_a\mathrel{\circ}\mathscr{R}$ is $\mathscr{S}$-stable.
\end{lemma}
\begin{proof}
  Since $\mathscr{S}$ is a simulation, we have
  $\mathscr{S}\mathrel{\circ}\pre_a\subseteq\pre_a\mathrel{\circ}\mathscr{S}$ and thus: 1)
  $\mathscr{S}\mathrel{\circ}\pre_a\mathrel{\circ}\mathscr{R}\subseteq\pre_a\mathrel{\circ}\mathscr{S}\mathrel{\circ}\mathscr{R}$. With
  the hypothesis that $\mathscr{R}$ is $\mathscr{S}$-stable, we get: 2)
  $\pre_a\mathrel{\circ}\mathscr{S}\mathrel{\circ}\mathscr{R}\subseteq\pre_a\mathrel{\circ}\mathscr{R}$. Inclusions
  1) and 2) put together imply the claimed property.
\end{proof}

\begin{proof}[Proof of Theorem~\ref{th:Refine}.]
  The fact that $\mathscr{U}$ is reflexive and block-definable is an easy
  consequence of its definition: from a reflexive and block-definable
  relation, $\mathscr{R}$, we split some blocks, then we delete some
  relations between \underline{different} blocks (
  after $\splitRefine$,
  we
  have: $[\pre_a(B)]_{\splitRefine_{a,\refiner}(\mathscr{R})}\cap
  [\remove_{a,\refiner}]_{\splitRefine_{a,\refiner}(\mathscr{R})}=\emptyset$).
  
  For the $\mathscr{S}$-stability of $\mathscr{U}$, let us first remark
  another direct consequence of the definitions of $\remove_{a,\refiner}$ and
  $\splitRefine_{a,\refiner}$:
  \begin{equation}
    \label{eq:remove-st-class}
    q \in \remove_{a,\refiner} \Rightarrow
    [q]_{\splitRefine_{a,\refiner}(\mathscr{R})} \subseteq
    \remove_{a,\refiner}
  \end{equation}

  If $\mathscr{U}$ is not $\mathscr{S}$-stable, there are three states
  $q_1, q_2, q_3 \in Q$ such that $q_1\,\mathscr{U}\,q_2 \wedge q_2
  \,\mathscr{S}\, q_3 \wedge \neg\; q_1\,\mathscr{U}\,q_3$.  We need the
  following property:
  \begin{equation}
    \label{eq:AppImpSstable}
    q_3 \in \remove_{a,\refiner} \Rightarrow
    q_2\not\in\pre_a(\mathscr{R}_1(B)) 
  \end{equation}
  Suppose $q_3\in\remove_{a,\refiner}$ and
  $q_2\in\pre_a(\mathscr{R}_1(B))$. Since $B$ is a block of $\mathscr{R}_1$
  and $\mathscr{R}_1$ is block-definable, for any $q\in B$ we have
  $q_3\in\mathscr{S}\mathrel{\circ}\pre_a\mathrel{\circ}\mathscr{R}_1(q)$. With the hypothesis
  that $\mathscr{R}_1$ is $\mathscr{S}$-stable, Lemma~\ref{lem:ImpSstable}
  implies $q_3\in \pre_a\mathrel{\circ}\mathscr{R}_1(q)$ which contradicts $q_3 \in
  \remove_{a,\refiner}$.
      
  Now, by construction, $\mathscr{U}$ is included in $\mathscr{R}$ and, by
  hypothesis, $\mathscr{R}$ is $\mathscr{S}$-stable, then from
  $q_1\,\mathscr{U}\,q_2 \wedge q_2 \,\mathscr{S}\, q_3$, and thus
  $q_1\,\mathscr{R}\,q_2 \wedge q_2 \,\mathscr{S}\, q_3$, we get:
  $q_1\,\mathscr{R}\,q_3$. With $ \neg\; q_1\,\mathscr{U}\,q_3$, we 
  necessarily have $(q_1,q_3)\in \splitDelete_{a,\refiner}(\mathscr{R})$ or
  $(q_1,q_3)\in \delete_{a,\refiner}(\mathscr{R})$.  Let us consider the two
  cases (also depicted in Figure~\ref{fig:splitRefinement}):

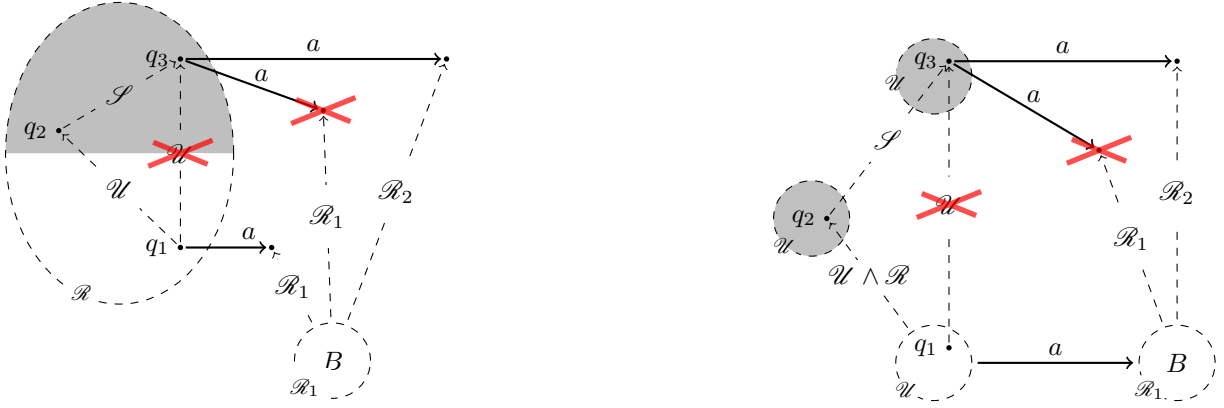
\begin{figure}[h]
  \centering
  \begin{tikzpicture}[shorten >=2pt, shorten <=2pt,font=\footnotesize]

    \node (B) [circle,draw,dashed,anchor=center,minimum size=1cm] {$B$} ++(90:4cm)
    +(1.5,0) coordinate (q3'1) circle (1pt) ++(-2,0) coordinate (q3)
    ++(-20:2) coordinate (q3'2) [fill] circle (1pt) (q3) ++(0,-2.5)
    coordinate (q1) ($(q1)!.5!(q3)$) ++(-.8,0) coordinate (mid13) (q1) ++(0:1.2)
    coordinate (q1') [fill] circle (1pt) (mid13) +(-.8,.3) coordinate (q2)
    ;

    \path (B.south west) node[fill=white,anchor=center,font=\scriptsize] {$\mathscr{R}_1$} ;

    \path[fill,black!25] (mid13) -- ++(1.5,0) arc (0:180:1.5cm and 2cm) --
    cycle ; \path[draw,dashed] (mid13) ellipse (1.5cm and 2cm) ;

    \path (mid13)
    ++(-100:1.7) node[anchor=north east,fill=white,inner sep=2pt,font=\scriptsize] {$\mathscr{R}$} ;
    
    \node[anchor=east] at (q3) {$q_3$} [fill] (q3) circle (1pt);
    \node[anchor=east] at (q2) {$q_2$} [fill] (q2) circle (1pt);
    \node[anchor=east] at (q1) {$q_1$} [fill] (q1) circle (1pt);
    
    \path[every edge/.style={->,dashed,draw},circle,inner
    sep=3pt,dashed,every node/.style={fill=white}] (B) edge node
    {$\mathscr{R}_2$} (q3'1) edge node {$\mathscr{R}_1$} (q3'2) edge node
    {$\mathscr{R}_1$} (q1') (q1) edge node {$\mathscr{U}$} (q2) (q2) edge
    node[fill=black!25] {$\mathscr{S}$} (q3) ;
         
    \path (q1) edge[->,dashed] node[circle] {} (q3); \fill[black!25]
    (mid13) ++(+0.6cm,0) rectangle +(.5cm,0.3); \fill[white] (mid13)
    ++(0.6cm,0) rectangle +(0.5cm,-0.15);

    \path (q1) edge[dash pattern=on 0pt off 3pt] node[circle] (q1Uq3)
    {$\mathscr{U}$} (q3);

    \path[->,circle,inner sep=1pt,thick,auto] (q3) edge node {$a$}
    (q3'1) edge node {$a$} (q3'2) (q1) edge node[near end] {$a$}
    (q1') ;

    \node[cross out,draw=red,line width=2pt,draw opacity=.70,text
    width=.5cm] at (q3'2) {}; \node [cross out,draw=red,line width=2pt,draw
    opacity=.70,text width=.5cm] at (q1Uq3) {};
  \end{tikzpicture}
  \hfill
  \begin{tikzpicture}[shorten >=2pt, shorten <=2pt,font=\footnotesize]
    \node (B) [circle,draw,dashed,minimum size=1cm] {$B$} ++(90:4cm)
    +(0,0) coordinate (q3'1) [fill] circle (1pt) ;
    \node[fill=white,circle,inner sep=0pt,minimum size=9pt,font=\scriptsize]
    at (B.225) {$\mathscr{R}_1$};

    \path (B) ++(90:4cm) ++(-3,0) coordinate (q3) ++(-.2,-.2)
    node[draw,dashed,circle,minimum size=1cm] (q3class) {};    
    \node[fill=white,circle,inner sep=0pt,minimum size=9pt,font=\scriptsize]
    at (q3class.190) {$\mathscr{U}$};
    \path  (q3class) node[fill,nearly transparent,circle,minimum size=1cm]{};    

    \path (q3) coordinate[label=left:\textcolor{black}{$q_3$}] [fill]
    circle (1pt) (B) ++(110:3) coordinate (q3'2) circle (1pt) ;

    \path let \p{q3}=(q3),\p{B}=(B.center) in (\x{q3},\y{B}) coordinate
    (q1') +(-.2,0) node[draw,dashed,circle,minimum
    size=1cm] (q1class) {} ;
    \node[fill=white,circle,inner sep=0pt,minimum size=9pt,font=\scriptsize]
    at (q1class.225) {$\mathscr{U}$};

    \path (q1')
    +(0,.2) coordinate [label=left:\textcolor{black}{$q_1$}] (q1) [fill]
    circle (1pt) ;

    \path (q1class) ++(130:2.5) node[draw,dashed,circle,minimum
    size=1cm] (q2class) {};
    \node[fill=white,circle,inner sep=0pt,minimum size=9pt,font=\scriptsize] at (q2class.225) {$\mathscr{U}$};
    \path  (q2class) node[fill,nearly transparent,circle,minimum size=1cm]{};    
    
    \path (q2class)    ++(.2,0) coordinate[label=left:{$q_2$}] (q2) [fill] circle (1pt);

    \path[every edge/.style={->,dashed,draw},circle,inner sep=3pt, every
    node/.style={fill=white}] (B) edge node {$\mathscr{R}_2$} (q3'1) edge
    node {$\mathscr{R}_1$} (q3'2) (q1class) edge node[rectangle,inner
    sep=2pt] {$\mathscr{U}\wedge\mathscr{R}$} (q2) (q1) edge node (q1Uq3)
    {$\mathscr{U}$} (q3) (q2) edge node {$\mathscr{S}$} (q3) ;
      
    \path[->,auto,circle,inner sep=1pt,thick] (q3) edge node {$a$} (q3'1)
    edge node {$a$} (q3'2) (q1class) edge node {$a$} (B) ;
        
    \node [cross out,draw=red,line width=2pt,draw opacity=.70,text
    width=.5cm] at (q3'2) {}; \node [cross out,draw=red,line width=2pt,draw
    opacity=.70,text width=.5cm] at (q1Uq3) {};
  \end{tikzpicture}

  \caption{$\mathscr{S}$-stability of the refinement}
  \label{fig:splitRefinement}
\end{figure}

\begin{itemize}
\item $(q_1,q_3)\in \splitDelete_{a,\refiner}(\mathscr{R})$.  This implies the
  existence of $q\in \remove_{a,\refiner}$ such that $q_1,q_3\in
  [q]_{\mathscr{R}}$, $q_1\not\in \remove_{a,\refiner}$ and $q_3\in
  \remove_{a,\refiner}$.  Since $\mathscr{R}$ is reflexive and
  $\mathscr{S}$-stable, we have $\mathscr{S}\, \subseteq \mathscr{R}$. With
  the fact that, by construction, $\mathscr{U}\, \subseteq \mathscr{R}$,
  from $q_1\,\mathscr{U}\,q_2 \wedge q_2 \,\mathscr{S}\, q_3$ we get
  $q_1\,\mathscr{R}\,q_2 \wedge q_2 \,\mathscr{R}\, q_3$. With $q_1,q_3\in
  [q]_{\mathscr{R}}$ and the fact that $\mathscr{R}$ is block-definable we
  get: $q_2\in [q]_{\mathscr{R}}$. With $q\in \pre_a(\mathscr{R}_2(B))$ and
  the fact that $(B,\mathscr{R}_1,\mathscr{R}_2)$ is a refiner of
  $\mathscr{R}$ we have, by definition: $q_2\in \pre_a(\mathscr{R}_1(B)\cup
  \mathscr{R}_2(B))$. With \eqref{eq:AppImpSstable} we necessarily have
  $q_2\in \pre_a(\mathscr{R}_2(B))\setminus\pre_a(\mathscr{R}_1(B)) =
  \remove_{a,\refiner} $. With $q_1\not\in \remove_{a,\refiner}$ and $q_1,q_2\in
  [q]_{\mathscr{R}}$ this would imply
  $(q_1,q_2)\in\splitDelete_{a,\refiner}(\mathscr{R})$ and would contradict the
  fact that $q_1\,\mathscr{U}\,q_2$.

\item $(q_1,q_3)\in \delete_{a,\refiner}(\mathscr{R})$.  This implies the
  existence of $q'_1\in \pre_a(B)$ and $q'_3\in \remove_{a,\refiner}$ such that
  $q_1\in [q'_1]_{\splitRefine_{a,\refiner}(\mathscr{R})}$ and $q_3\in
  [q'_3]_{\splitRefine_{a,\refiner}(\mathscr{R})}$.  From
  \eqref{eq:remove-st-class} we get $q_3\in \remove_{a,\refiner}$.  From
  $q_1\,\mathscr{R}\,q_2$, the fact that $\mathscr{R}$ is block-definable and
  $q_1\in [q'_1]_{\splitRefine_{a,\refiner}(\mathscr{R})}$, thus $q'_1\in
  [q_1]_{\mathscr{R}}$ since $\splitRefine_{a,\refiner}(\mathscr{R}) \subseteq
  \mathscr{R}$, we get $q'_1\,\mathscr{R}\,q_2$. With $q'_1\in \pre_a(B)$
  we have $q_2\in\mathscr{R}(\pre_a(B))$.  With the fact that
  $(B,\mathscr{R}_1,\mathscr{R}_2)$ is a refiner of $\mathscr{R}$ we
  have, by definition, $q_2\in \pre_a(\mathscr{R}_1(B)\cup
  \mathscr{R}_2(B))$.  With \eqref{eq:AppImpSstable} we necessarily have
  $q_2\in\remove_{a,\refiner}$.  With $q'_1\in \pre_a(B)$ and $q_1\in
  [q'_1]_{\splitRefine_{a,\refiner}(\mathscr{R})}$ this implies $(q_1,q_2)\in
  \delete_{a,\refiner}(\mathscr{R})$ and contradicts the fact that
  $q_1\,\mathscr{U}\,q_2$.
\end{itemize}
Both cases lead to a contradiction. The relation $\mathscr{U}$ is thus
$\mathscr{S}$-stable.

Let us now prove the last property:
\begin{itemize}
\item
  $[\pre_a(\mathscr{R}_1(B))]_{\mathscr{U}}\subseteq\pre_a(\mathscr{R}_1(B))$. Let
  $q\in[\pre_a(\mathscr{R}_1(B))]_{\mathscr{U}}$. There is
  $q'\in\pre_a(\mathscr{R}_1(B))$ such that $q\in[q']_{\mathscr{U}}$. Since
  $\mathscr{U}\subseteq\mathscr{R}$ and $(B,\mathscr{R}_1,\mathscr{R}_2)$
  is a refiner of $\mathscr{R}$ then $q\in\pre_a(\mathscr{R}_1(B)\cup
  \mathscr{R}_2(B))$. If $q\not\in\pre_a(\mathscr{R}_1(B))$ then $q\in
  \remove_{a,\refiner}$, which implies
  $(q',q)\in\splitDelete_{a,\refiner}(\mathscr{R})$ and contradicts
  $q'\,\mathscr{U}\,q$.
\item $\mathscr{U}(\pre_a(B))\subseteq\pre_a(\mathscr{R}_1(B))$. Let
  $q\in\mathscr{U}(\pre_a(B))$. There is $q'\in\pre_a(B)$ such that
  $q'\,\mathscr{U}\,q$. Since $\mathscr{U}\subseteq\mathscr{R}$ and
  $(B,\mathscr{R}_1,\mathscr{R}_2)$ is a refiner of $\mathscr{R}$
  then $q\in\pre_a(\mathscr{R}_1(B)\cup \mathscr{R}_2(B))$. If
  $q\not\in\pre_a(\mathscr{R}_1(B))$ then $q\in \remove_{a,\refiner}$, which
  implies $(q',q)\in\delete_{a,\refiner}(\mathscr{R})$ and contradicts
  $q'\,\mathscr{U}\,q$.
\end{itemize}
\end{proof}

\section{Base Algorithm}

  \begin{function}[H]
    \caption{Split($Remove,P$)\label{func:split}}

    $SplitCouples := \emptyset$; $Touched := \emptyset$; $BlocksInRemove := \emptyset$\;
    
    \ForAll{$r\in Remove$} {
      $Touched := Touched \cup \{r.\Block\}$\;
    }

    \ForAll{$C\in Touched $} {
      \If {$C\subseteq Remove$}
      {$BlocksInRemove := BlocksInRemove\cup\{C\}$\;}
      \Else(\ //$C$ must be split)
      {
        $D := C\cap Remove$; $P := P \cup \{D\}$\;
$BlocksInRemove = BlocksInRemove\cup\{D\}$\;        
        $C := C\setminus Remove$\; 
        {\ // Only $C$ is modified, not $C.\Rel$ or $C.\NotRel$}\nllabel{split:l11}\;
         $D.\Rel := \Copy(C.\Rel)$\nllabel{split:l12}\;
         $D.\NotRel := \Copy(C.\NotRel)$\nllabel{split:l13}\;
        \lForAll{$q\in D$} {q.\Block \,:= D}\;
        $SplitCouples := SplitCouples \cup \{(C,D)\}$\;
      }
     }

    \ForAll{$(C,D)\in SplitCouples,\,E\in P$\nllabel{split:l16}}
    {      
        \If{$C\in E.\Rel$} {$E.\Rel := E.\Rel \cup
          \{D\}$\nllabel{split:l18}\; }
      }

      \Return{$(P$, $BlocksInRemove$, $SplitCouples)$}
  \end{function}
  
  \begin{function}[H]
    \caption{Init($T,P_{init},R_{init}$) with $T={(Q,\Sigma,\rightarrow)}$\label{func:init}}   
    \BlankLine
    $P := \Copy(P_{init})$; $S := \emptyset$\nllabel{init:l1}\;
    
    \lForAll{$a\in\Sigma $\nllabel{init:l2}}
    {$a.\Remove :=\emptyset$;}

    \ForAll{$B\in P$\nllabel{init:l3}}
    {
      $B.\Rel := \{C \in P \suchthat (B,C) \in R_{init}\}$\;
    }

    \ForAll{$q\xrightarrow{a}q'\in\rightarrow$\nllabel{init:l5}}
    {                     
      $a.\Remove := a.\Remove \cup \{q\}$\nllabel{init:l6}\;
    }     
                        
    \ForAll{$a\in \Sigma$\nllabel{init:l7}}
    {                
      $(P,BlocksInRemove,\_) :=\Split(a.\Remove,P)$\nllabel{init:l8}\;      
      \ForAll{$C\in BlocksInRemove,\, D\in P$\nllabel{init:l9}}
      {
        \If{$D\not\in BlocksInRemove$}
        {
          $C.\Rel := C.\Rel \setminus \{D\}$\nllabel{init:l11}\;
        }                                
      }
    }
    \ForAll{$C\in P$ \nllabel{init:l12}}
    {
      $C.\NotRel :=\cup \{D\in P\suchthat D \not\in C.\Rel\}$\;
      \lIf{$C.\NotRel \neq \emptyset$}{$S:=S\cup\{C\}$\nllabel{init:l14};}
    }

    \Return{$(P$, $S)$}
  \end{function}
  
    \begin{function}
      \caption{Sim($T, P_{init}, R_{init}$) with $T={(Q,\Sigma,\rightarrow)}$}
      \label{func:sim}
   $ (P, S) := \Init(T,P_{init},R_{init})$\;
    \lForAll{$a\in \Sigma $}
    {\{$a.\PreB := \emptyset ; a.\Remove :=\emptyset$\nllabel{sim:l2}\}}\;
    $alph := \emptyset$\;          
   
    \While{$\exists B \in S$\nllabel{sim:l4}} {
      
      $S := S\setminus \{B\}$\nllabel{sim:l5}\;
      
      { // Assert : $alph = \emptyset\wedge(\forall a\in\Sigma\,.\,a.\PreB = \emptyset\wedge
        a.\Remove =\emptyset)$\nllabel{sim:l6}}

      \ForAll{$r\xrightarrow{a}(B.\NotRel)
        \wedge r \not\in \pre_a(\cup B.\Rel)$\nllabel{sim:l7}}
      {
        $alph := alph \cup \{a\}$\nllabel{sim:l8}\;
        $a.\Remove := a.\Remove \cup \{r\}$\nllabel{sim:l9}\;
      }
      $B.\NotRel :=\emptyset$\nllabel{sim:l10}\;
        \ForAll{$q\xrightarrow{a}B\wedge a\in alph$\nllabel{sim:l11}}
        {$a.\PreB := a.\PreB \cup \{q\}$\nllabel{sim:l12}\;}
      
      \ForAll{$a\in alph$ \nllabel{sim:l13}}
      {        
        $(P, BlocksInRemove, SplitCouples) :=
        \Split(a.\Remove,P)$\nllabel{sim:l14}\;
        
        \ForAll{$(C,D)\in SplitCouples$\nllabel{sim:l15}}
        { $C.\Rel := C.\Rel
          \setminus \{D\}$ ;
          $C.\NotRel := C.\NotRel \cup D$\nllabel{sim:l16}\;          
          $S := S \cup \{C\}$ \nllabel{sim:l17}\;}

        \ForAll{
          $\begin{array}[c]{ll}
           D\in BlocksInRemove, \\
           C\in \{q.\Block\in P\suchthat q\in a.\PreB\}
          \end{array}$\nllabel{sim:l18}
        }
        {
          \If{$D\in C.\Rel$}
          {
          $C.\Rel := C.\Rel \setminus \{D\}$; $
          C.\NotRel := C.\NotRel \cup D$\nllabel{sim:l20}\;
          $S := S \cup \{C\}$\nllabel{sim:l21}\;}
        }
      }
      \lForAll{$a\in alph $\nllabel{sim:l22}}{\{$a.\PreB := \emptyset ; a.\Remove
        :=\emptyset$\}}\;
        $alph := \emptyset$\nllabel{sim:l23}\;          
    }
    
    $P_{sim} := P$; $R_{sim}:= \{(B,C)\in P\times P\suchthat C\in
    B.\Rel\}$\;
    \Return{$(P_{sim}$, $R_{sim})$}    
  \end{function}

Given a LTS $T=(Q,\Sigma, \rightarrow)$ and an initial antisymmetric
partition-relation pair $(P_{init},R_{init})$, inducing a preorder
$\mathscr{R}_{init}$, the 
algorithm manipulates relevant refiners to 
iteratively refine $(P,R)$ initially set to $(P_{init},R_{init})$.
At the end, $(P,R)$ represents $(P_{sim},R_{sim})$ the partition-relation
pair whose induced relation $\mathscr{R}_{sim}$ is the coarsest simulation
included in $\mathscr{R}_{init}$.

The partition $P$ is a set of blocks. To represent $R$, we simply associate
to each block $B\in P$ a set $B.\Rel\subseteq P$ such that
$R=\cup_{B\in P}\{B\}\times B.\Rel$. A block is assimilated with its set of
states. For a given state $q\in Q$, the block of $P$ which contains $q$ is
noted $q.\Block$. We also associate to a block $B$ a set $B.\NotRel$
included in the complement of $\cup B.\Rel$. The refiners will be of
the form: $(B, B\times(\cup B.\Rel),B\times (B.\NotRel))$.

The algorithm is decomposed in three functions: \texttt{Split},
\texttt{Init} and \texttt{Sim}, the main one. The function \texttt{Split}$(Remove,P)$, 
used by the two others, splits, possibly, the blocks touched by
$\remove$ and returns the updated partition, the list of blocks included in
$\remove$ and the list of block couples issued 
from a split. This last list permits the $\splitRefine_{a,\refiner}$ of
Definition~\ref{def:refine} 
during the \textbf{forall} loop at line~\ref{sim:l15} of \texttt{Sim}. 
Note that, even if \texttt{Split} possibly modify the current
partition-relation pair, thanks to lines \ref{split:l12} and
\ref{split:l16}--\ref{split:l18} it does not modify the induced
relation. This is done out of \texttt{Split}, in \texttt{Sim}.

The main role of function \texttt{Init} is to transform the initial
partition-relation pair to a partition-relation pair whose induced
relation satisfies condition \eqref{eq:InitRefineRestriction}. It also
initializes the set $S$ to those $B$'s whose $B.\NotRel$ is not empty. 

After the initialization, the \texttt{Sim} function mainly executes the
following loop.  As long as there is a block $B$ whose
$B.\NotRel$ is not empty, all non empty $a.\Remove$ sets are
computed by only one (this is important for the time efficiency) scan of
the transitions leading into $B.\NotRel$. Each of them 
corresponds to a $\remove_{a,\refiner}$ of 
Definition~\ref{def:refine}. The relevant $\pre_a(B)$, encoded by
$a.\PreB$, are also computed by only one (idem) scan of
the transitions leading into $\pre_a(B)$. Then, for each letter, with a non empty
$a.\Remove$, a 
refinement step is executed with the refiner $(B, B\times(\cup B.\Rel),
B\times (B.\NotRel))$. Note that, during a refinement step, each time a
relation $(C,D)$ is removed from $R$, the content of $D$ is
added to $C.\NotRel$. This is done in order to preserve the second
invariant of Lemma~\ref{lem:mainInv}. 
The remainder of the section validates the algorithm.

\begin{lemma}
  \label{lem:InitAlgo}
  Let $T=(Q,\Sigma, \rightarrow)$ be a LTS and $(P_{init}, R_{init})$ be a
  partition-relation pair over $Q$ inducing a preorder $\mathscr{R}_{init}$. Let
  $(P, \_) = \Init(T,P_{init}, R_{init})$ and 
  $\mathscr{R}=\cup_{G\in P}G\times(\cup G.\Rel)$.
  Then, $\mathscr{R} = \initRefine(\mathscr{R}_{init})$. Furthermore, for all
  $G\in P$, we have: $G.\NotRel=Q\setminus \cup G.\Rel$.
\end{lemma}
\begin{proof}
  Unless otherwise specified, all line numbers refer to function
\texttt{Init}. 
The purpose of the \textbf{forall} loop at  line
\ref{init:l3} is to associate to each block
$B\in P$ the initial set of blocks which simulate it: 
$B.\Rel = {R}_{Init}(B)$.
In  the \textbf{forall} loop at  line \ref{init:l5} we
identify $\pre_a(Q)$, encoded by $a.\Remove$, for each letter $a\in \Sigma$. Then, in the
\textbf{forall} loop at line \ref{init:l7}, for each relevant letter
 $a\in\Sigma$:
 \begin{itemize}
 \item We split each block $B\in P$ in two parts, $B\cap \pre_a(Q)$ and
   $B\setminus\pre_a(Q)$, and
   we update $R$ such that the induced relation of $(P,R)$ stays the
   same (lines \ref{split:l11}--\ref{split:l12} and
   \ref{split:l16}--\ref{split:l18} of function \texttt{Split}).
 \item Now, each block of $P$ is either included in $\pre_a(Q)$ or disjoint
   from it. We then delete from $R$ all couple $(C,D)$ such that $C$ is
   included in $\pre_a(Q)$ and $D$ is disjoint from $\pre_a(Q)$.
 \end{itemize}
 At the end $\mathscr{R}$, the induced relation of $(P,R)$, is
 $\mathscr{R}_{init}$ where all couples $(q,q')$ such
 that $q\in\pre_a(Q)$ and $q'\not\in\pre_a(Q)$ have been deleted. Otherwise
 said  $\mathscr{R} = \initRefine(\mathscr{R}_{init})$. Then, the
 \textbf{forall} loop at line \ref{init:l12} implies
 $C.\NotRel=Q\setminus \cup C.\Rel$ for all node $C\in P$. 
\end{proof}

\begin{lemma}
  \label{lem:mainInv}  
  Let $T=(Q,\Sigma, \rightarrow)$ be a LTS, $(P_{init},R_{init})$ be an
  initial partition-relation pair over $Q$ inducing a preorder
  $\mathscr{R}_{init}$ and $\mathscr{S}$ be a simulation over $T$ such that
  $\mathscr{S}\subseteq \mathscr{R}_{init}$. Let $\mathscr{R}=\cup_{G\in
    P}G\times(\cup G.\Rel)$. Then, the following properties are invariants
  of the \texttt{\bf while} loop of function \texttt{Sim}:  
     \begin{enumerate}
     \item $\mathscr{R}$ is a reflexive, block-definable and $\mathscr{S}$-stable
        relation,
      \item $  
      \forall G\in P\; \forall c\in \Sigma\;.\;
        [\pre_c(\cup G.\Rel\cup G.\NotRel)]_{\mathscr{R}} \cup
        \mathscr{R}(\pre_c(G))
        \subseteq
        \pre_c(\cup G.\Rel\cup G.\NotRel)
      $
    \end{enumerate}  
\end{lemma}
\begin{proof}
  The proof is done by an induction on the iterations of the loop. The two
properties are 
true just after the initialization. For the first one,
from Lemmas~\ref{lem:InitRefine} and
\ref{lem:InitAlgo}, we deduce that $\mathscr{S}\subseteq\mathscr{R}$,
$\mathscr{R}$ is a preorder and thus reflexive and 
block-definable. With the fact that a preorder is, by
definition, transitive  we deduce that $\mathscr{S}\subseteq\mathscr{R}$
implies the $\mathscr{S}$-stability of $\mathscr{R}$. For the second one,
this is a direct consequence of item
\ref{lem:InitRefine:3} of Lemma~\ref{lem:InitRefine} and the fact that just after
the initialization:
$\cup G.\Rel\cup G.\NotRel=Q$, see Lemma~\ref{lem:InitAlgo}, for all block
$G\in P$.

Let us consider an iteration of the loop. For the ease of the
demonstration, we prime a variable for its
value before the iteration.
A value during the
iteration is not primed. The two properties are supposed true before the
iteration, we show they are still true after.
Therefore, we assume:
\begin{equation}
  \label{eq:mainHypothesis}
      \forall G\in P'\; \forall c\in \Sigma\;.\;
      \begin{array}[c]{c}
        [\pre_c(\cup G.\Rel'\cup G.\NotRel')]_{\mathscr{R}'} \cup\\
        \mathscr{R}'(\pre_c(G))\\
        \subseteq\\
        \pre_c(\cup G.\Rel'\cup G.\NotRel')
      \end{array}
\end{equation}
 
In this proof, all line numbers, if not stated otherwise, refer to
function \texttt{Sim}, and $B$ is the block considered at
line \ref{sim:l4}.
Let $\refiner=(B,\mathscr{R}_1,\mathscr{R}_2)$ with $\mathscr{R}_1=B\times
(\cup B.\Rel')$, $\mathscr{R}_2=B\times
  (B.\NotRel')$. Then, $\refiner$ is a refiner of $\mathscr{R}'$. This is due
  to the following facts:
  \begin{itemize}
  \item the reflexivity of $\mathscr{R}'$ implies that $B$ is a block of
    $\mathscr{R}_1$, 
  \item $\mathscr{R}'(B) = \mathscr{R}_1(B)\subseteq\mathscr{R}_1(B)$,
  \item the partitionability of $\mathscr{R}'$
    implies the  partitionability of
    $\mathscr{R}_1$,
  \item from \eqref{eq:mainHypothesis} we have:
      \begin{displaymath}
     \forall c\in\Sigma\;.\;
     [\pre_c(\mathscr{R}_1(B)\cup\mathscr{R}_2(B))]_{ \mathscr{R}'}
     \cup \mathscr{R}'(\pre_c(B))
     \subseteq
     \pre_c(\mathscr{R}_1(B)\cup \mathscr{R}_2(B))
   \end{displaymath}  
  \end{itemize}

Clearly, after the first iteration of the \textbf{forall} loop at line
\ref{sim:l13}, from Theorem~\ref{th:Refine}, we have
$\mathscr{R}=\refine_{a,\refiner}(\mathscr{R}')$ and 
$\mathscr{R}$ is reflexive, block-definable, $\mathscr{S}$-stable and
$\mathscr{R}\subseteq\mathscr{R}'$. Therefore, $\refiner$ is still a
refiner of $\mathscr{R}$. The same happens for the successive
iterations of the \textbf{forall} loop at line
\ref{sim:l13}. The first property of the current lemma is thus true.

To prove the second property of the lemma, we need two intermediate results.

\begin{equation}
  \label{eq:descOfNotBInduc}
\forall F\in P'\setminus\{B\}\ \forall G\in P\;.\; G\subseteq F \Rightarrow
\cup F.\Rel'\cup F.\NotRel'=\cup G.\Rel\cup G.\NotRel
\end{equation}
By an induction on the number of splits from $F$ to $G$. If there has been
no split then $G=F$ and only line \ref{sim:l20} can modify 
$G.\Rel$ or $G.\NotRel$, but such that the expression $\cup G.\Rel\cup
G.\NotRel$ stays constant. Suppose the property is true before a split.
If that split does not involve $G$ then, thanks to lines
\ref{split:l16}--\ref{split:l18} in \texttt{Split}, $\cup G.\Rel$ and
$G.\NotRel$ are not modified. If it is a split of
 $G$ in $G_1$ and $G_2$. Then, the split is done such that, function
 \texttt{Split} lines \ref{split:l11}--\ref{split:l13}, $\cup G_i.\Rel=\cup
 G.\Rel$, $G_i.\NotRel=G.\NotRel$ and, function \texttt{Split} lines
 \ref{split:l16}--\ref{split:l18}, $\mathscr{R}$ is not changed. With 
 the induction hypothesis  we get:
 $\cup F.\Rel'\cup F.\NotRel' = \cup G_i.\Rel\cup G_i.\NotRel$.
After the split, only lines \ref{sim:l16} and \ref{sim:l20} can modify 
 $G_i.\Rel$ or $G_i.\NotRel$, but such that the expression $\cup G_i.\Rel\cup
 G_i.\NotRel$ stays constant.

\begin{equation}
  \label{eq:descOfB}
\forall G\in P\;.\; G\subseteq B \Rightarrow
\cup B.\Rel'=\cup G.\Rel\cup G.\NotRel
\end{equation}
The proof is similar to the previous one except that $B.\NotRel'$ has been
emptied at line~\ref{sim:l10}.

Let $G\in P$ after a given iteration of the \textbf{forall} loop at line
\ref{sim:l13}. There are two cases:

\begin{itemize}
\item There is $F\in P'\setminus\{B\}$ such that $G\subseteq F$.
  From \eqref{eq:mainHypothesis} and
  \eqref{eq:descOfNotBInduc} we get: $ \forall c\in \Sigma\;.\;
    [\pre_c(\cup G.\Rel\cup G.\NotRel)]_{\mathscr{R}'} \cup
    \mathscr{R}'(\pre_c(G))
    \subseteq
    \pre_c(\cup G.\Rel\cup G.\NotRel)
  $.
  From the fact that $\mathscr{R}\subseteq\mathscr{R}'$ we obtain: $
  \forall c\in \Sigma\;.\;
    [\pre_c(\cup G.\Rel\cup G.\NotRel)]_{\mathscr{R}} \cup
    \mathscr{R}(\pre_c(G))
    \subseteq
    \pre_c(\cup G.\Rel\cup G.\NotRel)
  $.

\item $G\subseteq B$.
   We have two sub cases:
  \begin{itemize}
  \item $c\in alph$. Let $\mathscr{R}_c$ be the value of $\mathscr{R}$
    after the iteration of the \textbf{forall} loop at line \ref{sim:l13} with
    $a=c$. From what precede, remember that $\refiner$ is still a refiner of
    $\mathscr{R}$, and $\mathscr{R}$ and $\mathscr{R}_1$ are still
    $\mathscr{S}$-stable. Then, from Theorem~\ref{th:Refine} we get $ 
    [\pre_c(\mathscr{R}_1(B))]_{\mathscr{R}_c} \cup
    \mathscr{R}_c(\pre_c(B)) \subseteq \pre_c(\mathscr{R}_1(B))$. At the
    end of the iteration of the while loop, we obviously have
    $\mathscr{R}\subseteq\mathscr{R}_c$. With \eqref{eq:descOfB} and
    $G\subseteq B$ we obtain: $ [\pre_c(\cup G.\Rel \cup
    G.\NotRel)]_{\mathscr{R}} \cup \mathscr{R}(\pre_c(G)) \subseteq
    \pre_c(\cup G.\Rel \cup G.\NotRel))$.
  \item $c\not\in alph$. In that case, $c.\Remove=\emptyset$, thus
    $\pre_c(B.\NotRel')\subseteq \pre_c(\cup B.\Rel')$. With
    \eqref{eq:mainHypothesis} we get: $ [\pre_c(\cup
    B.\Rel')]_{\mathscr{R}'} \cup \mathscr{R}'(\pre_c(B)) \subseteq
    \pre_c(\cup B.\Rel')$. With \eqref{eq:descOfB},
    $\mathscr{R}\subseteq\mathscr{R}'$ and $G\subseteq B$ we obtain: $
    [\pre_c(\cup G.\Rel \cup G.\NotRel)]_{\mathscr{R}} \cup
    \mathscr{R}(\pre_c(G)) \subseteq \pre_c(\cup G.\Rel \cup G.\NotRel))$.
  \end{itemize}

\end{itemize}
\end{proof}

\begin{theorem}
  \label{thm:SimIsCorrect}
  Let $T=(Q,\Sigma, \rightarrow)$ be a LTS and $(P_{init},R_{init})$ be an
  initial partition-relation pair over $Q$ inducing a preorder
  $\mathscr{R}_{init}$. Function \texttt{Sim} computes the partition-relation pair
  $(P_{sim},R_{sim})$ inducing $\mathscr{R}_{sim}$ the maximal 
  simulation over $T$ contained in $\mathscr{R}_{init}$. Furthermore,
  $\mathscr{R}_{sim}$ is a preorder.
\end{theorem}
\begin{proof}
   From line \ref{init:l14}
    of function \texttt{Init}, lines \ref{sim:l17} and \ref{sim:l21} of
    function \texttt{Sim}, a block $G\in P$ is added in $S$ whenever
    $G.\NotRel$ is not empty. Furthermore, each time a block $G$ is
    withdrawn from $S$, line \ref{sim:l5}, $G.\NotRel$ is emptied, line
    \ref{sim:l10}. Therefore, a block $G$ is in $S$ iff $G.\NotRel$
    is not empty. 
  \begin{enumerate}
  \item \textbf{Function \texttt{Sim} terminates}. Let
    $\mathscr{R}'=\cup_{G\in P}G\times(\cup G.\Rel\cup G.\NotRel)$. At each
    iteration of the \texttt{\bf while} loop $\mathscr{R}'$ strictly
    decrease (a not empty $B.\NotRel$ is emptied). Since $\mathscr{R}'$ is
    a finite set the algorithm terminates necessarily.

  \item
    \textbf{$\mathscr{R}_{sim}$ is a simulation}.  The algorithm terminates when $S$ is
    empty. From what precede, at this moment, for all $G\in P$,
    $G.\NotRel=\emptyset$. With Lemma~\ref{lem:mainInv} we get:
    $  \forall G\in P\; \forall a\in \Sigma\;.\;
      \mathscr{R}_{sim}(\pre_a(G))
      \subseteq
      \pre_a( \mathscr{R}_{sim}(G))$.
    With Proposition \ref{prop:blockSim} this means that
    $\mathscr{R}_{sim}$ is a simulation since $\mathscr{R}_{sim}$ is
    reflexive and block-definable (Lemma~\ref{lem:mainInv}).

  \item \textbf{$\mathscr{R}_{sim}$ contains all simulation included in
      $\mathscr{R}_{init}$}.  From Lemma~\ref{lem:mainInv} we deduce that
    for all simulation $\mathscr{S}$ over $T$ included in
    $\mathscr{R}_{init}$, $\mathscr{R}_{sim}$ is reflexive and
    $\mathscr{S}$-stable. These two properties 
    imply $\mathscr{S}\subseteq\mathscr{R}_{sim}$.

  \item \textbf{$\mathscr{R}_{sim}$ is a preorder}.  We have seen that
    $\mathscr{R}_{sim}$ is $\mathscr{S}$-stable for any simulation included
    in $\mathscr{R}_{init}$. Since $\mathscr{R}_{sim}$ is such a
    simulation, $\mathscr{R}_{sim}$ is also $\mathscr{R}_{sim}$-stable and
    thus transitive. This relation being reflexive and transitive is a
    preorder.
  \end{enumerate}
\end{proof}

\section{Complexity}
\label{sec:complexity}

From now on, all space complexities are given in
bits.
Let $X$ be a set of elements, we qualify an encoding of $X$ as
\emph{indexed} if the elements of $X$ are encoded in an array of 
$|X|$ slots, one for each element. Therefore, an elements of $X$ can be
identify with its index in this array.
Let $T=(Q,\Sigma, \rightarrow)$ be a LTS, an encoding of $T$ is said
\emph{normalized} if the encodings of $Q$, $\Sigma$ and $\rightarrow$ are
indexed, a transition is encoded by the index of its source
state, the index of its label and the index of its destination state, and
if  $|Q|$ and $|\Sigma|$ are in $O(|{\rightarrow}|)$.
If $|\Sigma|$ is not in $O(|{\rightarrow}|)$, we can restrict it to its
useful part $\Sigma'=\{a\in\Sigma\suchthat \exists q,q'\in
Q\;.\;q\xrightarrow{a}q'\in\rightarrow\}$ whose size is less than 
$|{\rightarrow}|$. To do this, we can use hash table techniques, sort
the set $\rightarrow$ with the keys being the letters labelling the transitions,
or more efficiently use a similar technique of that we used in the algorithm to
distribute a set of transitions relatively to its labels (see, as an
example, the \textbf{forall} loop at line \ref{sim:l7} of
\texttt{Sim}). This is done in $O(|\Sigma|+|{\rightarrow}|)$ time and uses
$O(|\Sigma|.\log|\Sigma|)$ space. We have recently seen that
this may be done in $O(|{\rightarrow}|)$ time, still with
$O(|\Sigma|.\log|\Sigma|)$ space, by using a technique presented
in \cite{Val09}.
If $|Q|$ is not in $O(|{\rightarrow}|)$ this means there are states that
are not involved in any transition. In general, these states are
ignored. Indeed, any state can simulate them and they can not simulate any
state with an outgoing transition. Therefore, we can restrict $Q$ to its
useful part $\{q\in Q\suchthat \exists q'\in Q\;\exists a\in \Sigma\;.\;
q\xrightarrow{a}q'\in\rightarrow \vee \;q'\xrightarrow{a}q\in\rightarrow\}$
whose size is in $O(|{\rightarrow}|)$. This is done like for
$\Sigma$.
\begin{center}
  All encodings of LTS in this paper are assumed to be normalized.
\end{center}
  We also assume the encoding of the initial partition-relation pair
$(P_{init},R_{init})$ to be such that: the encoding of $P_{init}$ is
indexed, for each block $B\in P$ scanning of the states in $B$ can be
done in $O(|B|)$ time and scanning of $R_{init}(B)$ can be
done in $O(|P_{init}|)$ time. Furthermore, for each state $q\in Q$, we
assume set $q.\Block$ the block of $P_{init}$ to which $q$ belongs.

One difficulty concerning the data structures was to design an efficient 
encoding of the $\NotRel$'s avoiding the need to update them each time
a block is split. This led us to design an original encoding of $P$
(encoding which finally happens to be similar to the one designed in
\cite{Val09}). 
The first two versions of the algorithm essentially differ
by the implementation of the test $r \not\in \pre_a(\cup B.\Rel)$ at line
\ref{sim:l7}; the overall time complexity of the other parts of the
algorithm being in $O(|P_{sim}|.|{\rightarrow}|)$.

\subsection{Hypothesis}
\label{sec:hypothesis}

In this sub-section, we state the relevant complexity properties of the
data structures we use. We postpone on Section~\ref{sec:DataStruct} the
explanations on how we meet these properties.

During an initialization phase, we set the following:
\begin{itemize}
\item 
  For each transition $t=q\xrightarrow{a}q'$, we set $t.\Sl=q_a\in \SL(\rightarrow)$ the state-letter
  associated with $t$.
\item For each state-letter $q_a\in \SL(\rightarrow)$, we set $q_a.\State=q$ and
  $q_a.\Post=\{q\xrightarrow{a}q'\in\rightarrow\suchthat q'\in Q\}$ such
  that scanning the 
  transitions in $q_a.\Post$ is done in linear time.
\item For each state $q'\in Q$, we set
  $q'.\Pre=\{q\xrightarrow{a}q'\in \rightarrow \suchthat q\in Q,
  \;a\in\Sigma\}$ such that scanning the transitions of this set is done in
  linear time. 
\end{itemize}
This initialization phase is done in $O(|{\rightarrow}|)$
  time and uses $O(|{\rightarrow}|.\log|{\rightarrow}|)$ space.
  
The partition $P$ is encoded such that adding a new block is done in
constant amortized time and scanning the states of a block is done in
linear time. The encoding of $P$ uses $O(|P_{sim}|.\log|{\rightarrow}|)$
space. Note that the content of all the blocks uses $O(|Q|.\log|Q|)$
space.

The union $B.\NotRel$ of blocks that do not simulate a given $B\in P$ is
encoded such that resetting to $\emptyset$ is done in constant time and adding the content of a
block is done in constant amortized time (relatively to the number of added
blocks) while scanning the states present
in the union is done in
linear time of the number of states. The encoding of all $\NotRel$'s uses
$O(|P_{sim}|^2.\log|P_{sim}|)$ space.

The set of blocks, $B.\Rel$, that simulate a given $B\in P$ is
encoded such that membership test and removing of a
block are done in constant time while adding of a block is done in constant
amortized time (relatively to the size of $P_{sim}$). The encoding of all
$\Rel$'s is done in $O(|P_{sim}|^2)$ space.

The set $S$ of blocks to be treated by the main loop of the algorithm is
encoded such that the emptiness test and the extraction of one element (arbitrarily
chosen by the data structure) are done in constant
time, and adding an element in $S$ is done in constant amortized
time. The encoding of $S$ uses $O(|P_{sim}|.\log|{\rightarrow}|)$ space.

The sets $alph$, $SplitCouples$, $Touched$ and $BlocksInRemove$ are encoded
such that adding of an element is done in constant time and,
scanning of their elements and resetting to $\emptyset$ are done in linear
time (relatively to the number of elements). The encoding of these sets
uses $O(|{\rightarrow}|.\log|{\rightarrow}|)$ space.

For all $a\in\Sigma$, $a.\PreB$ and $a.\Remove$ are encoded such that 
adding of an element is done in constant time,
scanning of their elements and resetting to $\emptyset$ is done in linear
time (relatively to the number of elements). The encoding of all $a.\PreB$
and $a.\Remove$ takes $O(|{\rightarrow}|.\log|{\rightarrow}|)$ space.

Finally,  $\Split(Remove,P)$ is done in  $O(|Remove|)$ time.

 \subsection{Common analysis}
 \label{sec:commonParts}

\begin{figure}[h]
  \centering
  \begin{tikzpicture}[shorten >=2pt, shorten <=2pt,font=\footnotesize]
    \path coordinate [label=right:\textcolor{black}{$q'$}]  (q')  [fill] circle (1pt) ;
    
      \path ++(0.2,0) node (B) [circle,draw,label={right:$B$},minimum size=1cm] {}
      ++(90:4cm)  +(1.5,0) coordinate (r') [fill] circle (1pt) ;

    \path (r') coordinate[label=right:\textcolor{black}{$r'$}]
            +(.2,0)  node[draw,circle,minimum size=1cm,
            label={right:$E$},
            label={above:$E\subseteq B.\NotRel$}] (r'class) {} ;
  
    \path (B) ++(90:4cm) ++(-3,0) coordinate (r) ++(-.2,0)
    node[fill,nearly transparent,draw,circle,minimum size=1cm,
    label={left:$D$},
    label={above:$D\subseteq a.\Remove$}] (rclass) {};
    \path (rclass) node[draw,circle,minimum size=1cm] {} ;

    \path (r) coordinate[label=left:\textcolor{black}{$r$}] [fill]
    circle (1pt) ;
    \path  ++(0.2,2.9) coordinate[label=above:\textcolor{black}{$r''$}] (r'') [fill] circle (1pt)
 node[draw,circle,minimum size=1cm] (r''class) {} 
    ;

    \path let \p{r}=(r),\p{B}=(B.center) in (\x{r},\y{B}) coordinate
    (q'') +(-.2,0) node[draw,circle,minimum
    size=1cm,label={[black]left:$C$}] (qclass) {} ; \path (q'')
    coordinate [label=left:\textcolor{black}{$q$}] (q) [fill]
    circle (1pt) ;
 
    \path[every edge/.style={->,dashed,draw},circle,inner sep=3pt, every node/.style={fill=white}]
    (B)
        edge node {$\mathscr{R}$} (r''class)
    ;
      
    \path[->,auto,circle,inner sep=1pt,thick]
    (r) edge node {$a$} (r')
        edge node {$a$} (r'')
    (q) edge node {$a$} (q') ;
        
    \node [cross out,draw=red,line width=2pt,draw opacity=.70,text
    width=1cm] at (r'') {};
  \end{tikzpicture}
    \caption{A configuration during a iteration of the while loop}
\label{fig:configIt}
\end{figure}
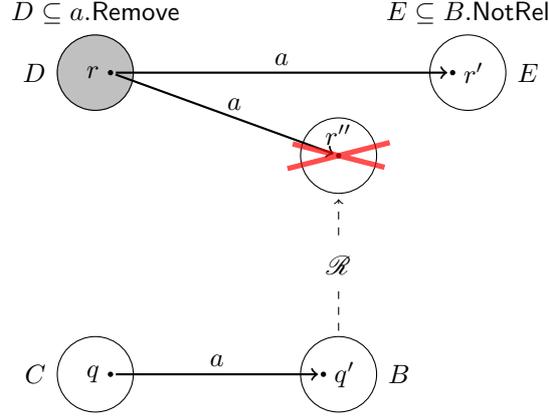

 Figure~\ref{fig:configIt} illustrates the main lemma of this section. 

\begin{lemma}
  \label{lem:AtMostOnce}
  Let $B$ be a block defined at line~\ref{sim:l4} of \emph{\texttt{Sim}}. 
  During the execution of \emph{\texttt{Sim}}, the following configurations
  can happen at most once at line~\ref{sim:l4} (line~\ref{sim:l18} for the
  last one): 
  \begin{enumerate}
  \item \label{lem:AtMostOnce:1} $B$ and a block $E$ such that 
  $E\subseteq B.\NotRel$,
\item \label{lem:AtMostOnce:2} $B$ and a transition
  $r\xrightarrow{a}r'$ such that $r'\in B.\NotRel$, 
\item \label{lem:AtMostOnce:3} $B$, a block $E\subseteq B.\NotRel$ and a
  transition $q\xrightarrow{a}q'$ such that $q'\in B$, 
\item \label{lem:AtMostOnce:4} $B$, a block $D$ and a transition $q\xrightarrow{a}q'$
    such that $D\subseteq a.\Remove$ (or $D\in BlocksInRemove$) and $q'\in B$.
  \end{enumerate}
\end{lemma}
\begin{proof}
  \mbox{}
  
  \begin{enumerate}
  \item After the initialization, the content of a block can be added into
    $B.\NotRel$ only if this block is removed from $\cup B.\Rel$,
    lines~\ref{sim:l16} and \ref{sim:l20}. Furthermore, $\cup B.\Rel$ can
    only decrease and if $E$ is included in $B.\NotRel$ at
    line~\ref{sim:l4}, $B.\NotRel$ is emptied at line~\ref{sim:l10}. From
    what precedes, it will not be possible again for $E$ to be included in
    $B.\NotRel$ during another iteration of the while loop.
  \item[2,3.] Direct consequences of the preceding point.
  \item[4.] Let us suppose this can happen twice. Let $B.\Rel'$ be the
    value of $B.\Rel$ the first time it happens and $B.Rel''$, $B.\NotRel''$ be the
    values of $B.\Rel$ and $B.\NotRel$ the second time it happens. With a same
    reasoning than that of the first point, we get: $B.\NotRel''\subseteq
    \cup B.\Rel'$. Let $r$ be any element of $D$. The first time the
    configuration happens, we necessarily have, lines~\ref{sim:l6},
    \ref{sim:l7} and \ref{sim:l9}: $r\not\in\pre_a(\cup B.\Rel')$. The
    second time the configuration happens we necessarily have:
    $r\in\pre_a(B.\NotRel'')\subseteq \pre_a(\cup B.\Rel')$ which
    contradicts $r\not\in\pre_a(\cup B.\Rel')$.
  \end{enumerate}
\end{proof}

\subsubsection{Time}
\label{sec:CommonTime}

In this sub-section, the $O$ notation refers to time complexity and the
\emph{overall complexity} of a line is the time complexity of 
all the executions of that line during the lifetime of the algorithm. 
\paragraph{Initialisation}
\label{sec:initialisation}

In this paragraph, we consider the complexity of the initialization. All
line numbers refer to function \texttt{Init}. 

Line~\ref{init:l1} essentially corresponds to a copy of $P_{init}$, this is
done in $O(|P_{init}|)$ and thus $O(|P_{sim}|)$.
Line~\ref{init:l2} is done in $O(|\Sigma|)$ and thus in $O(|{\rightarrow}|)$.
The \textbf{forall} loop at line~\ref{init:l3} is done in $O(|P_{init}|^2)$
and thus in $O(|P_{sim}|^2)$.  
The \textbf{forall} loop at line~\ref{init:l5} is done in
$O(|{\rightarrow}|)$.
Since we have $\Sigma_{a\in\Sigma}|a.\Remove|\leq |\SL(\rightarrow)|\leq
|{\rightarrow}|$, the overall complexity of line~\ref{init:l8} is $O(|{\rightarrow}|)$.
At first glance the overall complexity of
lines~\ref{init:l9}--\ref{init:l11} is in $O(|\Sigma|.|P_{sim}|^2)$ but it is
also in $O(|\SL(\rightarrow)|.|P_{sim}|)$, since there is at most $|\SL(\rightarrow)|$ blocks
$C$ concerned by $BlocksInRemove$ and each time it is for a given $a\in \Sigma$,
and thus in $O(|{\rightarrow}|.|P_{sim}|)$.
The \textbf{forall} loop at line~\ref{init:l12} is done in
$O(|P_{sim}|^2)$. From all of that, the complexity of \texttt{Init} is 
$O(|{\rightarrow}|.|P_{sim}|)$.

\paragraph{Simulation algorithm}
\label{sec:simulationAlgo}
Thanks to Lemma~\ref{lem:AtMostOnce}, item~\ref{lem:AtMostOnce:1}, the
\textbf{while} loop of function \texttt{Sim} is executed at most
$|P_{sim}|^2$ times since a block $B$ is concerned by the loop only if
$B.\NotRel\neq\emptyset$ and $B.\NotRel$ is made of a union of blocks.

The first two versions of the algorithm differ by the test $r\not\in\pre_a(\cup
B.\Rel)$ at line~\ref{sim:l7}. Therefore, in this paragraph, we do not consider
the overall complexity of this test. We only consider right now the overall
complexity of the scanning of all transitions $r\xrightarrow{a}r'$ such
that $r'\in\pre_a(B.\NotRel)$ and the overall complexity of
lines~\ref{sim:l8}--\ref{sim:l9}. From Lemma~\ref{lem:AtMostOnce}
item~\ref{lem:AtMostOnce:2} it is $O(|P_{sim}|.|{\rightarrow}|)$.

From Lemma~\ref{lem:AtMostOnce}
item~\ref{lem:AtMostOnce:3},  the overall complexity of the 
\texttt{forall} loop at line~\ref{sim:l11} is 
$O(|P_{sim}|.|{\rightarrow}|)$.

The two preceding paragraphs imply that the overall complexity of
lines~\ref{sim:l22} and \ref{sim:l23} is 
$O(|P_{sim}|.|{\rightarrow}|)$ since resetting $a.\PreB$  or $a.\Remove$ is
linear in their sizes and thus less than what has been added in them, which
is less than $O(|P_{sim}|.|{\rightarrow}|)$ (overall complexity of lines
\ref{sim:l12} and \ref{sim:l9}).

The overall complexity of line~\ref{sim:l13} is less than that of
line~\ref{sim:l8} and thus $O(|P_{sim}|.|{\rightarrow}|)$.

The complexity of \texttt{Split($a.\Remove,P$)} is 
$O(|a.\Remove|)$. From the time complexity of line~\ref{sim:l9} we get the
overall complexity of line~\ref{sim:l14}: $O(|P_{sim}|.|{\rightarrow}|)$.

There is at most $|P_{sim}|$ couples $(C,D)$ issued from the splits of
blocks. So, the overall complexity of the \textbf{forall} loop at
line~\ref{sim:l15} is $O(|P_{sim}|)$.

From the overall complexity of lines~\ref{sim:l9} and \ref{sim:l12}, the
overall complexity of the calculation of all $D$ and of all $C$ concerned
by line~\ref{sim:l18} is $O(|P_{sim}|.|{\rightarrow}|)$. From
Lemma~\ref{lem:AtMostOnce}  item~\ref{lem:AtMostOnce:4}, there is at most
$O(|P_{sim}|.|{\rightarrow}|)$ couples $(C,D)$ which have been involved at
line~\ref{sim:l18}. This implies the overall time complexity of the
\textbf{forall} loop of line~\ref{sim:l18}: $O(|P_{sim}|.|{\rightarrow}|)$.

With what precedes, the test $r\not\in\pre_a(\cup B.\Rel)$ at
line~\ref{sim:l7} being aside, the overall time complexities of the other lines
of the algorithm are all in $O(|P_{sim}|.|{\rightarrow}|)$.

\subsubsection{Space}
\label{sec:CommonSpace}

Apart from the data structures needed to do the test $r\not\in\pre_a(\cup B.\Rel)$ at
line~\ref{sim:l7}, from Section~\ref{sec:hypothesis}, the space complexity
of the algorithm is 
$O(|P_{sim}|^2.\log|P_{sim}| +|{\rightarrow}|.\log|{\rightarrow}|)$.

\subsection{The nice compromise}
\label{sec:compromise}

We use a set of state-letters, $SL\subseteq \SL(\rightarrow)$, with the same time and space
complexity properties as those of $alph$. Before line~\ref{sim:l7}
this set is emptied. Then, let us consider a given $r'\in
B.\NotRel$. From $r'$ we 
get, in linear time, all $t=r\xrightarrow{a}r'\in r'.\Pre$. If $r_a= t.\Sl$
is already in $SL$ it has already been treated, so we stop there and
consider the next element of 
$r'.\Pre$. Otherwise, we have not yet tested whether $r\not\in\pre_a(\cup
B.\Rel)$. To do that, first, we add $r_a$ into $SL$ and then, consider all
$r\xrightarrow{a}r''\in r_a.\Post$. If for none of them $r''.\Block$ is in
$B.\Rel$, which is tested each time in constant time, then
$r\not\in\pre_a(\cup B.\Rel)$. Thanks to the use of $SL$, from
Lemma~\ref{lem:AtMostOnce}, item~\ref{lem:AtMostOnce:1}, a transition
$r\xrightarrow{a}r''\in r_a.\Post$ is considered only once for a given couple
$(B,E)$ of blocks in Figure~\ref{fig:configIt}. Therefore, the overall time
complexity of the test $r\not\in\pre_a(\cup B.\Rel)$ in line~\ref{sim:l7}
is $O(|P_{sim}|^2.|{\rightarrow}|)$.

We can also express the time
complexity in another way. For that, we need to introduce the state-letter
branching factor of a LTS. The \emph{branching factor of a state} is the
number of its outgoing transitions. The \emph{branching factor of a
  state-letter} $q_a$ is the number of the outgoing transitions of $q$
labelled by $a$. The \emph{state branching factor} of a LTS is the greatest
branching factor of its states. The \emph{state-letter branching factor} of
a LTS is the greatest branching factor of its state-letters. Let us come
back to the analysis of the complexity of the test
$r\xrightarrow{a}(B.\NotRel) \wedge r \not\in \pre_a(\cup B.\Rel)$. From
Lemma~\ref{lem:AtMostOnce} item~\ref{lem:AtMostOnce:2}, a configuration
such that $r\xrightarrow{a}r'$ is a transition with $r'\in B.\NotRel$
happens only once. From this configuration, and with the use of the set
$SL$ described above, we have to consider at most $b$ transitions
$r\xrightarrow{a}r''\in r_a.\Post$ to test whether $r\not\in\pre_a(\cup
B.\Rel)$.

From what precedes we obtain the following theorem.

\begin{theorem}
  \label{th:niceCompromiseVersionComplexities}
  Let $T=(Q,\Sigma, \rightarrow)$ be a LTS and $(P_{init},R_{init})$ be an
  initial partition-relation pair over $Q$ inducing a preorder
  $\mathscr{R}_{init}$. The nice compromise version of \texttt{\em Sim} computes the
  partition-relation pair $(P_{sim},R_{sim})$ inducing $\mathscr{R}_{sim}$
  the maximal simulation over $T$ contained in $\mathscr{R}_{init}$ in:
  \begin{itemize}
  \item $O(\min(|P_{sim}|,b).|P_{sim}|.|{\rightarrow}|)$ time, and
  \item $O(|P_{sim}|^2.\log|P_{sim}|+|{\rightarrow}|.\log|{\rightarrow}|)$
    space.
  \end{itemize}
  with
  $b=\max_{q_a\in \SL(\rightarrow)}|
  \{q\xrightarrow{a}q'\in\rightarrow\suchthat q'\in Q\}|$
  the \emph{state-letter branching factor} of $T$. 
\end{theorem}

 Clearly, the state-letter branching factor of a LTS is
smaller than its state branching factor. For the state-letter branching
factor $b$ in the preceding theorem, we have: $b\leq |Q|$. We also
have $|P_{sim}|\leq |Q|$. But there 
is no definitive comparison between $b$ and $|P_{sim}|$. However if one
considers the VLTS Benchmark Suite
(\url{http://cadp.inria.fr/resources/benchmark_bcg.html}), the state
branching factor is rarely more than one hundred even for systems with
millions of states. Furthermore, in the case of deterministic systems, the
state-letter branching factor is 1. Therefore, we consider this version of
\texttt{Sim} as a nice compromise between space and time efficiency.

\subsection{The Time Efficient Version}
\label{sec:timeVersion}

To get an efficient time version of the algorithm, we need counters. To
each block $B\in P$ we associate $B.\RelCount$, an array of counters
indexed on the set of state-letters $\SL(\rightarrow)$ such that:
$B.\RelCount(r_a)=|\{r\xrightarrow{a}r'\in\rightarrow\suchthat r'\in \cup
B.\Rel\cup B.\NotRel\}|$. Let $r_a.\Post=\{r\xrightarrow{a}r'\in\rightarrow\suchthat r'\in Q\}$.
The initialization consists just of setting
$B.\RelCount(r_a)=|r_a.\Post|$ since at the end of the initialization, for
any $B\in P$, $Q=\cup B.\Rel\cup B.\NotRel$. Therefore, the time complexity
of the initialization of all the counters is 
$O(|P_{sim}|.|\SL(\rightarrow)|)$ and thus $O(|P_{sim}|.|{\rightarrow}|)$.

Let us come back to the overall time complexity of line~\ref{sim:l7}. For
each transition $r\xrightarrow{a}r'$ with $r'\in B.\NotRel$,
$B.\RelCount(r_a)$ is decremented, and if after that $B.\RelCount(r_a)=0$
this implies that $r\not\in\pre_a(\cup B.\Rel)$. This means that the
test $r\not\in\pre_a(\cup B.\Rel)$ is done in constant time for each
transition $r\xrightarrow{a}r'$ with $r'\in B.\NotRel$.
Note that when a block is split during the function \texttt{Split} its
array of counters must be copied to the new block. This is done in
$O(|\SL(\rightarrow)|)$. Since during the computation there is at most
$|P_{sim}|$ splits, the overall time complexity of all the copy operations
is  $O(|P_{sim}|.|\SL(\rightarrow)|)$ and thus
$O(|P_{sim}|.|{\rightarrow}|)$. We then get the following theorem.

\begin{theorem}
  \label{th:timeEfficientVersionComplexities}
  Let $T=(Q,\Sigma, \rightarrow)$ be a LTS and $(P_{init},R_{init})$ be an
  initial partition-relation pair over $Q$ inducing a preorder
  $\mathscr{R}_{init}$. The time efficient version of \texttt{\em Sim}
  computes the partition-relation pair 
  $(P_{sim},R_{sim})$ inducing $\mathscr{R}_{sim}$ the maximal
  simulation over $T$ contained in $\mathscr{R}_{init}$ in:
  \begin{center}
   $O(|P_{sim}|.|{\rightarrow}|)$ time and 
   $O(|P_{sim}|.|\SL(\rightarrow)|.\log|Q|+|{\rightarrow}|.\log|{\rightarrow}|)$ space.
  \end{center}
\end{theorem}

\subsection{The Space Efficient Version}
\label{sec:SpaceEfficientVersions}

The algorithm GPP has a time complexity of
$O(|P_{sim}|^2. |{\rightarrow}|)$, for Kripke structures,
but an announced space complexity of $O(|P_{sim}|^2 +
|Q|.\log|P_{sim}|)$. Unfortunately, this space complexity does not
correspond to that of GPP. As announced in the
introduction of the present paper, GPP uses (a modified version of)
HHK. For each state $q'\in Q$, HHK uses an array of
counters, to speed up the algorithm, and a set of states, $\remove(q')$, that
do not lead via a transition to a state simulating $q'$. The counters and
the $\remove$ sets use $O(|Q|^2.\log|Q|)$ bits. 
As GPP uses HHK on an abstract structure whose states correspond to
blocks of the current partition, the initial version of GPP uses
$O(|P_{sim}|^2.\log|P_{sim}|)$ bits for the counters and the $\remove$
sets. Then, the authors explain how 
to avoid the use of the counters, but do not do the same for the
$\remove$ sets. Therefore their algorithm still uses at least
$O(|P_{sim}|^2.\log|P_{sim}|)$ bits. The $O(|Q|.\log|P_{sim}|)$
part of their announced space complexity comes from the necessity to
memorize for each state $q$ the block to which it belongs ($q.\Block$ in
the present paper). But GPP, like the algorithm in \cite{CRT11}, scan in
linear time the 
states belonging to a block. To do that the set of the states of a block is
encoded by a doubly linked list which also enable to remove and to add a
state in a block in constant time. This implies that the size of each pointer of these
lists need to be sufficient to distinguish the $|Q|$ states: $\log
|Q|$. Since there is $|Q|$ states, GPP, like the algorithm in \cite{CRT11}, needs  
at least $|Q|.\log|Q|$ bits. The real space complexity of
GPP is thus $O(|P_{sim}|^2.\log|P_{sim}|+|Q|.\log|Q|)$.

By removing the use of the $\NotRel$'s in our base algorithm we are able to
propose the space efficient version. The time complexities of \texttt{Init} and
\texttt{Split} do not change, but now the overall time
complexity of almost all the lines in the \textbf{while} loop of
\texttt{Sim} becomes $O(|P_{sim}|^2.|{\rightarrow}|)$.

We now present how to avoid, in our nice compromise version of \texttt{Sim},
the use of the $\NotRel$'s in order to lower the space complexity from
$O(|P_{sim}|^2.\log|P_{sim}| +|{\rightarrow}|.\log|{\rightarrow}|)$ to 
$O(|P_{sim}|^2+|{\rightarrow}|.\log|{\rightarrow}|)$ while keeping a time
complexity of $O(|P_{sim}|^2.|{\rightarrow}|)$.

The transformation of the algorithm is quite simple: we mainly replace
line~\ref{init:l14} of \texttt{Init} by ``$S:=\Copy(P)$;'', 
line~\ref{sim:l7} of \texttt{Sim} by
``\textbf{forall} $r\xrightarrow{a}(\overline{\cup B.\Rel})
\wedge r \not\in \pre_a(\cup B.\Rel)$ \textbf{do}'',
and we remove all other instructions where $\NotRel$ appears. 
Just to simplify the proof, we also replace
line~\ref{sim:l17} of \texttt{Sim} by ``$S:=S\cup\{C,D\};$''.

In the remainder of this sub-section, lines refer only to the function
\texttt{Sim}. 

\subparagraph{Correctness}
\label{sec:correctness}
Clearly, after the initialization, which puts us under condition
\eqref{eq:InitRefineRestriction}, for any block $B\in P$, $(B, B\times
(\cup B.\Rel),B\times \overline{\cup B.\Rel})$ is a refiner of the
current relation. By noting $B.\NotRel\triangleq\overline{\cup B.\Rel}$,
second item of Lemma~\ref{lem:mainInv} becomes trivial. For the first item,
we follow the same proof. The algorithm terminates since, after the
first time, a block $B$
can be chosen again by line~\ref{sim:l4} only if it is issued from a split
(new line~\ref{sim:l17}) or if a block has been removed from $B.\Rel$
(line~\ref{sim:l21}). Each case can happen at most $|P_{sim}|$ times. Let
$\mathscr{R}_{sim}$ be the relation induced by the result
$(P_{sim},R_{sim})$ of \texttt{Sim}. Like 
in the proof of Theorem~\ref{thm:SimIsCorrect}, we use
Lemma~\ref{lem:mainInv} to 
deduce that $\mathscr{R}_{sim}$ contains all simulation included in
$\mathscr{R}_{init}$ the relation induced by the initial partition-relation
pair. Now, for a given block $B$ of the last partition, consider the last
time $B$ has been chosen by line~\ref{sim:l4}. After, the execution of the
corresponding iteration of the loop, from Theorem~\ref{th:Refine} we deduce
that $\mathscr{R}(\pre_a(B))\subseteq\pre_a(\cup B.\Rel')$ with $B.\Rel'$
the value of $B.\Rel$ before the iteration. But since
this is the last use of $B$, $B.\Rel$ has not been modified during this
iteration of the while loop. Thus, $B.\Rel=B.\Rel'$. From this moment on,
$B.\Rel$ will not be modified. So we have 
$\mathscr{R}_{sim}(\pre_a(B))\subseteq\pre_a(\mathscr{R}_{sim}(B))$ for
all block $B\in P_{sim}$ and all letter $a\in\Sigma$. This defines a
simulation (Proposition~\ref{prop:blockSim}).

\subparagraph{Complexity}

The \textbf{forall} loop at line~\ref{sim:l7} is encoded by the following
lines: 

{    \restylealgo{plain}
    \linesnotnumbered

    \begin{center}
      \begin{algorithm}[H]
        \ForAll{$q\xrightarrow{a}q'\in \rightarrow$} { \lIf{$q'.\Block\in
            B.\Rel$} {$a.\PreRel:=a.\PreRel\cup\{q\}$;} }
        \ForAll{$q\xrightarrow{a}q'\in \rightarrow$} {
          \If{$q'.\Block\not\in B.\Rel\wedge q\not\in a.\PreRel$} { $alph
            := alph \cup \{a\}$\; $a.\Remove := a.\Remove \cup \{q\}$\; } }
         \end{algorithm}
    \end{center}
}

In addition, we add ``$a.\PreRel:=\emptyset$;'' to the bodies of
lines \ref{sim:l2} and \ref{sim:l22}. The $a.\PreRel$'s are data structures with the same
complexity properties as those of $a.\PreB$ and $a.\Remove$. For a given
iteration of the \textbf{while} loop the time complexity of these lines is
$O(|{\rightarrow}|)$. Since the number of iterations of the \textbf{while}
loop is in $O(|P_{sim}|^2)$, the overall time complexity of these lines is
$O(|P_{sim}|^2.|{\rightarrow}|)$. This is thus also the overall time complexity
of lines
\ref{sim:l13}, \ref{sim:l14}, \ref{sim:l22} and \ref{sim:l23}, and the
overall time complexity of calculation of all $C$ and $D$ concerned by
line~\ref{sim:l18}. The 
overall time complexity of the \textbf{forall} loop at line~\ref{sim:l15}
does not change: $O(|P_{sim}|)$. Consider now Figure~\ref{fig:configIt}. A
transition $q\xrightarrow{a}q'$ with $q'$ in a block $B$ chosen at
line~\ref{sim:l4} is considered only $O(|P_{sim}|)$ times. Knowing that for
each time there is at most $|P_{sim}|$ blocks $D$ in $a.\Remove$ we deduce
the overall time complexity of the \textbf{forall} loop at
line~\ref{sim:l18}: $O(|P_{sim}|^2.|{\rightarrow}|)$.

Note for the attentive reader: continuing the discussion just before
Definition~\ref{def:refine}, in practice, it should be more interesting
to do the split with $\pre_a(\cup B.\Rel)$ instead of $a.\Remove$ and to do
the refinement step with $BlocksInRemove=\{D\in P\suchthat D\cap\pre_a(\cup
B.\Rel)=\emptyset\}$. Because, $\pre_a(\cup B.\Rel)$ is supposed to
decrease at each iteration which is not the case for $a.\Remove$ when we no
longer have $B.\NotRel$.

\begin{theorem}
  \label{th:spaceEfficientVersionComplexities}
  Let $T=(Q,\Sigma, \rightarrow)$ be a LTS and $(P_{init},R_{init})$ be an
  initial partition-relation pair over $Q$ inducing a preorder
  $\mathscr{R}_{init}$. The space efficient version of
  \texttt{\em Sim} 
  computes the partition-relation pair 
  $(P_{sim},R_{sim})$ inducing $\mathscr{R}_{sim}$ the maximal forward
  simulation on $T$ contained in $\mathscr{R}_{init}$ in:
  \begin{center}
   $O(|P_{sim}|^2.|{\rightarrow}|)$ time and 
   $O(|P_{sim}|^2+|{\rightarrow}|.\log|{\rightarrow}|)$ space.
  \end{center}
\end{theorem}

In the case of Kripke structures, it seems possible to derive a version of
the algorithm which still works in $O(|P_{sim}|^2.|{\rightarrow}|)$ time
but uses only $O(|P_{sim}|^2+|Q|.\log|P_{sim}|)$ space. The idea is to not 
use the option to scan the states of a block in linear time. The split
operation is thus now done in $O(|Q|)$ time and so one. This is just a bit
tedious to do. Moreover, in
practical cases, the problem is not the
$O(|{\rightarrow}|.\log|{\rightarrow}|)$ part of our space complexity 
but the
$O(|P_{sim}|^2)$ part.

\section{Data structures}
\label{sec:DataStruct}

In what follows, we use different kinds of data: simple objects, arrays and
lists of objects. The size of the pointers has an importance for bit space
complexity. It should be 
enough to differentiate all the considered objects and thus
$O(|{\rightarrow}|)$ for normalized LTS. We call a resizable array an array
which double its size when needed. Therefore, adding a new item in this kind
of array is done in constant amortized time. 

First, we have to set $t.\Sl$ for each transition $t\in\rightarrow$, $q_a.\State$ and
$q_a.\Post$ for each state-letter $q_a$. To do
that we create a new array of transitions, $Post$, as the result of sorting the set
of transitions with the labels as keys, then with the source states as
keys. We use counting sorts. This means that the two sorts are done in
$O(|{\rightarrow}|)$ time since $Q$ and $\Sigma$ are in
$O(|{\rightarrow}|)$. Counting sorts are 
stable. As a result, in $Post$, 
transitions are packed by source states and within a pack of transition
sharing the same source state, there are the sub-packs of transitions sharing
the same label. Then, we scan the elements of $Post$ from the first one to
the last one. For each transition $t=q\xrightarrow{a}q'$, we consider the
couple $(q,a)$ and whenever it changes we create a new state-letter $q_a$
and we set $t.\Sl = q_a$. Then, we set $q_a.\State=q$ and
$q_a.\Range=(idx_{start}, idx_{end})$ with $idx_{start}$ the index in
$Post$ of the first transition from $q$ and labelled by $a$, and
$idx_{end}$ the index in $Post$ of the last transition from $q$ labelled by
$a$. Thanks to the two sorts, $(q_a.\Range, Post)$ provides an
encoding of $q_a.\Post$.
To represent $q'.\Pre$, we just create a new array of transitions, $Pre$, as
the result of sorting, by a counting sort, the set of transitions with
destination states as  keys. Then, by scanning this array, we associate to
each state $q'$ the couple $q'.\Range=(idx_{start}, idx_{end})$ with
$idx_{start}$ the index of the first transition in $Pre$ having $q'$ as
destination state and $idx_{end}$ the index of the last transition in $Pre$
having $q'$ as destination state. Therefore,  $(q'.\Range, Pre)$ provides an
encoding of $q'.\Pre$. All of this is done in
$O(|{\rightarrow}|)$ time and uses
$O(|{\rightarrow}|.\log|{\rightarrow}|)$ space.

We do not encode the content of a block of $P$ by a doubly-linked list of
states like the other papers because we need a certain stability
property. The problem is the following. 
let $C$ be a block, we want to be able to add the content of another block
$D$ in $C.\NotRel$ in constant time, without scanning the content of
$D$. We also want the encoding of $C.\NotRel$ to be in
$O(|P_{sim}|.\log|P_{sim}|)$ space. A first idea is to store in $C.\NotRel$ only the
reference of $D$. But a problem arises when $D$ is split in $D_1$ and
$D_2$: we have to update all the $C.\NotRel$ and replace $D$ by $D_1$ and
$D_2$; this implies an overall time complexity of $O(|P_{sim}|^3)$ since
there may have $|P_{sim}|$ splits. But $O(|P_{sim}|^3)$ is
too much for the time efficient version of our algorithm.
A solution is to use a kind of
family tree. A block of $P$ is a leave of the tree and is linked to the set
of its states.  When a block is split, it becomes an internal node of the
tree and is no more directly linked to the set of
its states but to its two son blocks (which are both linked to their
respective set of states). This solution satisfy all the
requirements since in a binary tree the number of nodes is at most two
times the number of leaves. 

However here is a more efficient solution. The set of states $Q$ is copied 
in a new array $Q_p$ such that 
the set
of states of a block $B\in P$ is arranged in consecutive slots of this array
$Q_p$. Therefore, to memorize the content of $B$, we just have to
memorize $B.\Start$ the starting position and $B.\End$ the ending position of the
corresponding 
subarray in $Q_p$. When a block  $B$ is split in two sub-blocks
$B_1$ and $B_2$, we just arrange the content of the subarray of $B$ in
$Q_p$ such that the first slots are for $B_1$ and the last slots are
for $B_2$. This way, even after several splits, the set of states
which once corresponded to a block of $P$ will always be in the same
subarray, even if the order of the states is modified. Note that to
do the rearrangement, during a split, of the states of $B$ in $Q_p$  we
need to memorize for a given state $r\in Q$ its position, $r.\PosQp$, in
$Q_p$. See function \texttt{SplitImplementation}.

For a given $C\in P$, $C.\NotRel$ may thus be encoded as a set of couples
$(x,y)$, which once corresponded to blocks of $P$, such that $x$ is the
start of a subarray in $Q_p$ and $y$ is the end of 
that subarray. Due to the fact that $B.\NotRel\cap(\cup
B.\Rel)=\emptyset$ and when the content of a block is added in $B.\NotRel$
the block is removed
from $B.\Rel$, all the blocks encoded in $C.\NotRel$ are
different. Therefore, $|C.\NotRel|$ is in $O(|P_{sim}|)$ and thus
the encoding of all the $\NotRel$'s is done in $O(|P_{sim}|^2.\log(|Q|)$
space. The factor $\log(|Q|)$ is due to the fact that in a couple $(x,y)$,
the maximum value for $x$ and $y$ is $|Q|$. However, we want the encoding of
all the $\NotRel$'s to be in $O(|P_{sim}|^2.\log(|P_{sim}|)$. To do that,
remember the family tree mentioned above. The number of past and actual
blocks of $P$ is in $O(|P_{sim}|)$. Therefore, we introduce $N$, a set of
nodes. During the initialization, we associate $B.\Node\in N$, a node, to each
block $B\in P$. The starting and ending position in $Q_p$ corresponding
to a block $B$ is not directly store in $B$ but in $B.\Node$ via
$B.\Node.\Start$ and $B.\Node.\End$. Therefore, when we want to add the
content of a block $D$ in $C.\NotRel$, in fact we add $D.\Node$ in
$C.\NotRel$. Since $|N|$ is in $O(|P_{sim}|)$, the encoding of all the
$\NotRel$'s is done in $O(|P_{sim}|^2.\log(|P_{sim}|)$ space. Note that the
encoding of all the nodes in $N$ is done in $O(|P_{sim}|.\log(|Q|)$
space. The $O(\log(|Q|)$ factor being for the couple $(n.\Start,n.\End)$
for each $n\in N$. The set $N$ and the $\NotRel$'s are encoded by resizable
arrays. 

For a given $B\in P$, the set $B.\Rel$ is encoded by a resizable boolean  
array. To know whether a block $C$
belongs to $B.\Rel$ we check $B.\Rel[C.\Index]$ with $C.\Index$ the index of
$C$ in the array encoding $P$. The encoding of all $\Rel$'s is therefore
done in $O(|P_{sim}|^2)$ space.

A given block $B\in P$ is encoded in $O(\log(|{\rightarrow}|)$ space since
we just need a constant number of integers, less than $|{\rightarrow}|$, or
pointers for $B.\Index$,
$B.\Node$,  $B.\NotRel$, $B.\Rel$, $B.\SplitCount$ (see function
\texttt{SplitImple\-mentation}) and $B.\RelCount$ (for the time efficient
version).  Thanks to, $Q_p$, $B.\Node.\Start$ and $B.\Node.\End$ scanning of the
states contained in a block  $B\in P$ is done in linear time. The set $P$
is encoded as a resizable array of blocks. Therefore, the  encoding of
$P$ is  done in $O(|P_{sim}|.\log(|{\rightarrow}|)$ space and the encoding
of the contents of the blocks of $P$ is done in $O(|Q|.\log(|Q|)$.

The set $S$ is encoded as a list of
blocks (we could have used a resizable array) but we also need to add a
boolean mark to the blocks of $P$ to know whether a given block is already
in $S$. That way, we keep the encoding of $S$ in
$O(|P_{sim}|.\log(|{\rightarrow}|))$ space.

The sets $alph$, $SplitCouples$ and $Touched$ are implemented like $S$: a list
and a binary mark on the respective elements. To reset one of these sets,
we simply scan the list of elements; for each of them we unset the
corresponding mark, then we empty the list. All of this is done in linear
time. The maximum sizes for $alph$ is $|\Sigma|$, for $SplitCouples$ and
$Touched$ it is $|P_{sim}|$. Therefore, they are all encoded in
$O(|{\rightarrow}|.\log(|{\rightarrow}|))$ space.

 To represent a set $a.\PreB$ or $a.\Remove$ with $a\in\Sigma$ we should not
  use a list of states and a binary array indexed on $|Q|$. This would have
  implied a total size of $|\Sigma|.|Q|$ for all the letters, which may
  exceed $|{\rightarrow}|$. Instead, we use a list of elements of
  $\SL(\rightarrow)$ per letter and only one common (for all the letters)
  binary array indexed on $\SL(\rightarrow)$. We also use the fact that for a
  given $a\in\Sigma$ a state can not belongs to both $a.\PreB$ and
  $a.\Remove$ in an iteration of the \textbf{while} loop of \texttt{Sim}. When we need to
  add a state $r$ in $a.\Remove$, for example, it is from a transition
  $r\xrightarrow{a}r'$ issued from a call of $r'.\Pre$. This call provides
   $r_a$ too. Then, we add $r_a$ in the list of $a.\Remove$ and we set the
   mark associate with $r_a$, only if this mark is not already set. Cleaning
   of $a.\Remove$ is done like cleaning of $alph$ (scanning the
   elements and unsetting the associated marks). Note that we store
   $r_a$ instead of $r$ in $a.\Remove$, but this is not a problem since
   $r_a.\State$ gives us $r$. The encoding of all $a.\PreB$ and $a.\Remove$
   is done in $O(|\SL(\rightarrow)|.\log |{\rightarrow}|)$ space and thus
   in $O(|{\rightarrow}|.\log |{\rightarrow}|)$ space. 

   As denoted by the name, function \texttt{SplitImplementation} is an
   implementation of function \texttt{Split} taking into account the new way
   of encoding the partition. Clearly, a call of
   \texttt{SplitImplementation($Remove,P$)} is done in $O(|Remove|)$ time.

  \begin{function}[h]
    \caption{SplitImplementation($Remove,P$)\label{func:splitImpl}}

    $SplitCouples := \emptyset$; $Touched := \emptyset$; $BlocksInRemove := \emptyset$\;
       { // Assert : $\forall C\in P\,.\,C.\SplitCount=0$}\;
       { // When a block is created, all its counters are set to 0.}\;
   
    \ForAll{$r\in Remove$} {
      $C:=r.\Block$\;
      $Touched := Touched \cup \{C\}$\;
      $oldpos:=r.\PosQp$; $newpos:=C.\Node.\Start+C.\SplitCount$\;
      $r':=Q_p[newpos]$\;
      $Q_p[newpos]:=r$; $Q_p[oldpos]:=r'$\;
      $r.\PosQp:=newpos$; $r'.\PosQp:=oldpos$\;
      $C.\SplitCount:=C.\SplitCount+1$ \;
    }

    \ForAll{$C\in Touched $} { %
      \If  {$C.\SplitCount= |C|$}
      {
        $BlocksInRemove:= BlocksInRemove\cup \{C\}$\;
      }
      \Else(\ //$C$ must be splitted)
      {
        $D:=\Newblock$; $P := P \cup \{D\}$\;
        $D.\Node:=\Newnode$; $N := N \cup \{D.\Node\}$\;        
        $BlocksInRemove:= BlocksInRemove\cup \{D\}$\;
        $D.\Node.\Start:=C.\Node.\Start$\;
        $D.\Node.\End:=C.\Node.\Start+C.\SplitCount-1$\;
        $C.\Node.\Start:=D.\Node.\End+1$\;
        $D.\Rel := \Copy(C.\Rel)$\;
        $D.\NotRel := \Copy(C.\NotRel)$\;
        $SplitCouples := SplitCouples \cup \{(C,D)\}$\;
 \lForAll{$pos\in \{D.\Node.\Start,\ldots,D.\Node.\End\}$} {$Q_p[pos].\Block:=D$}\;              
     }
     $C.\SplitCount:=0$\;      
     
}
    \ForAll{$(C,D)\in SplitCouples,\,E\in P$}
    {      
        \If{$C\in E.\Rel$} {$E.\Rel := E.\Rel \cup
          \{D\}$\; }
      }

      \Return{$(P$, $BlocksInRemove$, $SplitCouples)$}  
  \end{function}

\section{Future Works}
\label{sec:future}

In order to simplify the presentation, no practical optimization  has been
proposed. This will be done in a future work with the implementation of
the algorithms. For the moment we just recall
an easy theoretical optimization: the coarsest bisimulation relation should
be computed before, and used by the algorithms computing the coarsest
simulation relation. This reduces $\SL(\rightarrow)$, which is really
important for the space complexity of the time efficient version of the
algorithm,  
and also reduces the transition relation, which has a positive impact on
the time complexity of all the versions of the algorithm.

Concerning the search of the coarsest bisimulation relation in a LTS, the
framework presented in the present paper can be adapted. We have
recently learned that an algorithm avoiding the effect of the size of the
alphabet in the time and space complexities of the bisimulation problem has
already been presented by Valmari \cite{Val09} in 2009. The approach of Valmari is different. 
His splitters (roughly speaking, they play the same role of our refiners
but are
adapted for the bisimulation problem) depend conceptually on letters but he uses two partitions of
the set of transitions, beside the classical one for the states, to avoid
the negative effect of the size of the alphabet.
At first glance, an adaptation of our present work in the case of
bisimulation yields a simpler algorithm than the one of Valmari and,
furthermore, closer to the one of Paige and Tajan for Kripke structures
\cite{PT87}. This will be made precise in a future paper.

\bibliography{simulation}

\end{document}